\documentclass[aps,pra,superscriptaddress,twocolumn,preprintnumbers,floatfix]{revtex4-2}
\usepackage{xcolor,graphicx}
\usepackage{amsfonts,amssymb,amsmath,amsthm}
\usepackage{microtype,mathrsfs,bm,bbm}
\usepackage[english]{babel}
\usepackage{comment}
\usepackage[bookmarksnumbered,linktocpage,hypertexnames=false,hidelinks,breaklinks=true,pdfusetitle,colorlinks=true,allcolors=baselRed]{hyperref}
\usepackage[capitalize]{cleveref}
\usepackage{mleftright,mathtools,thmtools}

\newtheorem{theorem}{Theorem}\crefname{theorem}{Theorem}{Theorems}

\usepackage{mfirstuc}
\newcommand{\newthm}[2]{
	\newtheorem{#1}{\makefirstuc{#1}}
	\crefname{#1}{\makefirstuc{#1}}{\makefirstuc{#2}}
}
\newthm{corollary}{corollaries}
\newcommand{\nwth}[1]{\newthm{#1}{{#1}s}}
\nwth{lemma}
\nwth{proposition}
\nwth{conjecture}
\theoremstyle{definition}
\nwth{definition}
\nwth{observation}

\DeclareMathOperator{\tr}{tr}
\DeclareMathOperator{\poly}{poly}

\DeclareMathOperator{\supp}{supp}

\newcommand{\1}{\mathbbm{1}}

\newcommand{\NN}{\mathbb{N}}
\newcommand{\RR}{\mathbb{R}}

\renewcommand{\P}{\mathcal{P}}
\renewcommand{\S}{\mathcal{S}}

\newcommand{\ket}[1]{|#1\rangle}
\newcommand{\bra}[1]{\langle #1|}
\newcommand{\proj}[1]{\ket{#1}\!\bra{#1}}
\newcommand{\dyad}[1]{\proj{#1}}
\newcommand{\braket}[1]{\langle #1\rangle}
\newcommand{\dagg}{^\dagger}

\DeclarePairedDelimiter{\norm}{\lVert}{\rVert}
\DeclarePairedDelimiter{\abs}{\lvert}{\rvert}
\DeclarePairedDelimiter{\ceil}{\lceil}{\rceil}
\DeclarePairedDelimiter{\floor}{\lfloor}{\rfloor}

\usepackage{tikz}
\usetikzlibrary{decorations.pathreplacing,calligraphy,decorations.markings,decorations.pathmorphing,math,fadings}

\definecolor{gibbsbox}{HTML}{FFE9BB}    % Figure 1
\definecolor{ebdgibbsbox}{HTML}{FFE9BB} % Equation (11)

\definecolor{baselMint}{RGB}{165,215,210}
\definecolor{baselLightMint}{RGB}{192,227,223}
\definecolor{baselRed}{RGB}{210,005,055}
\definecolor{baselGrey}{RGB}{45,55,60}
\definecolor{baselLightGrey}{RGB}{70,80,90}
\colorlet{tensor}{baselLightMint}
\colorlet{X}{black!10}

\newcommand{\cfGap}{.15}% inset of a straight box edge from the half-site line
\tikzfading[name=cfhalo,inner color=transparent!0,outer color=transparent!100]
\tikzset{cfzig/.style={decoration={zigzag,segment length=2pt,amplitude=.45pt,
	pre length=0pt,post length=0pt}},
	cfzigthin/.style={decoration={zigzag,segment length=1pt,amplitude=.35pt,
	pre length=0pt,post length=0pt}}}

\def\cfDelta{.45}
\def\cfCa{1000}\def\cfCb{1000}\def\cfCc{1000}% last site before each cut
\newcommand{\cfOff}[1]{\cfDelta*(((#1)>\cfCa)+((#1)>\cfCb)+((#1)>\cfCc))}
\newcommand{\cfBox}[7]{%
	\pgfmathsetmacro\cfya{#4}\pgfmathsetmacro\cfyb{#5}%
	\pgfmathsetmacro\cfo{\cfOff{#2}}%
	\pgfmathparse{\cfyb-\cfya<.15?"cfzigthin":"cfzig"}\let\cfZ\pgfmathresult
	\if z#6\pgfmathsetmacro\cfxl{#2-.5+\cfo}\def\cfEdgeL{decorate[style=\cfZ]{--(\cfxl,\cfyb)}}%
	\else\pgfmathsetmacro\cfxl{#2-.5+\cfGap+\cfo}\def\cfEdgeL{--(\cfxl,\cfyb)}\fi
	\if z#7\pgfmathsetmacro\cfxr{#3+.5+\cfo}\def\cfEdgeR{decorate[style=\cfZ]{--(\cfxr,\cfya)}}%
	\else\pgfmathsetmacro\cfxr{#3+.5-\cfGap+\cfo}\def\cfEdgeR{--(\cfxr,\cfya)}\fi
	\filldraw[fill=#1,line join=round] (\cfxl,\cfya) \cfEdgeL --(\cfxr,\cfyb) \cfEdgeR --cycle;
}
\newcommand{\cfWires}[3]{\foreach \i in {1,...,#1}{%
	\pgfmathsetmacro\cfo{\i+\cfOff{\i}}%
	\draw[line width=.4pt] (\cfo,#2)--(\cfo,#3);}}
\newcommand{\cfKs}[1]{\foreach \j in {#1}{%
	\draw[line width=.4pt] (5*\j+1,\cfWt)--(5*\j+1,\cfWt+.075);
	\node[anchor=south,inner sep=1pt] at (5*\j+1,\cfWt+.075) {$k_\j$};}}

\def\cfSh{.09}% strip height in panel 3
\def\cfXh{.3}% X height in panel 2 (room for labels)
\def\cfGa{.68}\def\cfGb{.98}% Gibbs-state row
\def\cfMa{1.03}\def\cfMb{1.33}% M row
\def\cfPa{.62}\def\cfPb{1.04}% taller Gibbs row of panel 1 (room for the label)
\def\cfMc{1.41}% taller M row of panel 2 (room for the labels)
\def\cfWt{1.48}% wire end
\def\cfTs{1.4}% wire spacing of the top row, relative to the bottom row

\newcommand{\cfLplain}[2]{#1_{#2}}
\newcommand{\cfLdag}[2]{#1_{#2}\mkern-2mu\raisebox{.75ex}{$\scriptscriptstyle\dagger$}}

\newcommand{\cfPanelOne}{%
	\cfBox{gibbsbox}{1}{12}{\cfPa}{\cfPb}{z}{z}
	\node[font=\scriptsize] at (6.5,{(\cfPa+\cfPb)/2}) {$e^{-\beta H/2}$};
}
\newcommand{\cfPanelTwo}[1]{%
	\cfBox{X}{1}{18}{0}{\cfXh}{z}{z}
	\cfBox{X}{2}{18}{\cfXh}{2*\cfXh}{s}{z}
	\node at (9.5,{\cfXh/2}) {$#1{X}{0}$};
	\node at (10,{1.5*\cfXh}) {$#1{X}{1}$};
	\cfBox{gibbsbox}{2}{5}{\cfGa}{\cfGb}{s}{s}
	\cfBox{gibbsbox}{7}{18}{\cfGa}{\cfGb}{s}{z}
	\cfBox{tensor}{1}{3}{\cfMa}{\cfMc}{z}{s}
	\cfBox{tensor}{4}{8}{\cfMa}{\cfMc}{s}{s}
	\node at (1.85,{(\cfMa+\cfMc)/2}) {$#1{M}{0}$};
	\node at (6,{(\cfMa+\cfMc)/2}) {$#1{M}{1}$};
}
\newcommand{\cfPanelThree}{%
	\cfBox{X}{1}{30}{0}{\cfSh}{z}{z}
	\foreach \j in {1,...,6}{\cfBox{X}{5*\j-3}{30}{\j*\cfSh}{(\j+1)*\cfSh}{s}{z}}
	\foreach \j in {1,...,6}{\cfBox{gibbsbox}{5*\j-3}{5*\j}{\cfGa}{\cfGb}{s}{s}}
	\cfBox{tensor}{1}{3}{\cfMa}{\cfMb}{z}{s}
	\foreach \j in {1,...,5}{\cfBox{tensor}{5*\j-1}{5*\j+3}{\cfMa}{\cfMb}{s}{s}}
	\cfBox{tensor}{29}{30}{\cfMa}{\cfMb}{s}{z}
}

\colorlet{cfblock}{gibbsbox!40!X}
\def\cfWa{.21}\def\cfTa{.26}\def\cfTb{.56}\def\cfTw{.71}
\newcommand{\cfMrow}{% drawn at y in [\cfTa,\cfTb]
	\cfBox{tensor}{1}{3}{\cfTa}{\cfTb}{z}{s}
	\cfBox{tensor}{44}{45}{\cfTa}{\cfTb}{s}{z}
	\foreach \j in {1,...,8}{\cfBox{tensor}{5*\j-1}{5*\j+3}{\cfTa}{\cfTb}{s}{s}}
}
\newcommand{\cfTerm}[4]{%
	\pgfmathtruncatemacro\cfCa{5*#1-2}%
	\pgfmathtruncatemacro\cfCb{5*#2-2}%
	\pgfmathtruncatemacro\cfCc{5*#3-2}%
	\cfWires{45}{-\cfTw}{\cfTw}
	\cfMrow
	\begin{scope}[yscale=-1]\cfMrow\end{scope}
	\cfBox{cfblock}{1}{1}{-\cfWa}{\cfWa}{z}{s}
	\foreach \j in {1,...,9}{%
		\pgfmathtruncatemacro\cfa{5*\j-3}%
		\pgfmathtruncatemacro\cfb{min(5*\j+1,45)}% block 9 ends at N, then wraps
		\pgfmathtruncatemacro\cfcut{(\j==#1)||(\j==#2)||(\j==#3)}%
		\ifnum\j=9 \def\cfR{z}\else\def\cfR{s}\fi
		\ifnum\cfcut=1
			\cfBox{cfblock}{\cfa}{\cfa+1}{-\cfWa}{\cfWa}{s}{s}
			\cfBox{cfblock}{\cfa+2}{\cfb}{-\cfWa}{\cfWa}{s}{\cfR}
		\else
			\cfBox{cfblock}{\cfa}{\cfb}{-\cfWa}{\cfWa}{s}{\cfR}
		\fi
	}
	\pgfmathsetmacro\cfxa{5*#1-1.5+\cfDelta*.5}%
	\pgfmathsetmacro\cfxb{5*#2-1.5+\cfDelta*1.5}%
	\pgfmathsetmacro\cfxc{5*#3-1.5+\cfDelta*2.5}%
	\foreach \x in {\cfxa,\cfxb,\cfxc}{%
		\draw[baselRed,line width=1pt] (\x,{-\cfTw-.38})--(\x,\cfTb);}
	\def\cfAy{-\cfTw-.23}%
	\begin{scope}[every node/.style={font=\scriptsize,text=baselRed}]
		\node at ({(\cfxa+\cfxb)/2},\cfAy) {$A_1^{(#4)}$};
		\node at ({(\cfxb+\cfxc)/2},\cfAy) {$A_2^{(#4)}$};
		\node at ({(\cfxc+45.5+3*\cfDelta)/2},\cfAy) {$A_3^{(#4)}$};
	\end{scope}
	\foreach \j in {0,...,8}{%
		\pgfmathsetmacro\cfo{5*\j+1+\cfOff{5*\j+1}}%
		\draw[line width=.4pt] (\cfo,\cfTw)--(\cfo,\cfTw+.075);}
	\foreach \s/\a/\b in {1/1.5/16.5,2/16.5/31.5,3/31.5/45.5}{%
		\draw[decorate,decoration={brace,amplitude=4pt,raise=0pt}]
			({\a+.15+\cfDelta*(\s-1)},\cfTw+.14)--({\b-.15+\cfDelta*\s},\cfTw+.14)
			node[midway,above=4pt,font=\scriptsize] {$\S_\s$};}
}

\newcommand{\constructionFigure}{%
	\begin{tikzpicture}[x=.1565cm,y=1cm,line width=.5pt,font=\tiny]
		\begin{scope}[shift={(4.8,0)},xscale=\cfTs]
			\cfWires{12}{-\cfWt}{\cfWt}
			\cfPanelOne
			\begin{scope}[yscale=-1]\cfPanelOne\end{scope}
		\end{scope}
		\node[font=\normalsize] at (26,0) {$=$};
		\begin{scope}[shift={(28.95,0)},xscale=\cfTs]
			\cfWires{18}{-\cfWt}{\cfWt}
			\cfPanelTwo{\cfLplain}
			\begin{scope}[yscale=-1]\cfPanelTwo{\cfLdag}\end{scope}
			\cfKs{0,1}
		\end{scope}
		\node[font=\normalsize] at (58.85,0) {$=$};
		\begin{scope}[shift={(62.25,0)},xscale=\cfTs]
			\cfWires{30}{-\cfWt}{\cfWt}
			\cfPanelThree
			\begin{scope}[yscale=-1]\cfPanelThree\end{scope}
			\cfKs{0,...,5}
			\foreach \r in {.42,.3}{\fill[X,path fading=cfhalo] (17,0) circle[radius=\r cm];}
			\node[font=\scriptsize] at (17,0) {$Y$};
			\draw[decorate,decoration={brace,amplitude=5pt}] (31.2,\cfGb)--(31.2,-\cfGb);
			\node[font=\scriptsize] at (33.3,0) {$\Omega$};
		\end{scope}
		\foreach \x/\l in {4.6/a,28.4/b,61/c}{%
			\node[anchor=north east,inner sep=1pt,font=\small] at (\x,\cfWt+.075+.2) {(\l)};}
		\node[anchor=south east,inner sep=1pt,font=\small] at (4.6,{-3.11+\cfTw+.075}) {(d)};
		\begin{scope}[shift={(0,-3.11)}]
			\node[font=\normalsize,anchor=east] at (3.9,0) {$=$};
			\begin{scope}[shift={(5,0)}]\cfTerm{2}{5}{8}{1}\end{scope}
			\node[font=\normalsize] at (55.2,0) {$+$};
			\begin{scope}[shift={(58.1,0)}]\cfTerm{1}{6}{8}{2}\end{scope}
			\node[font=\normalsize,anchor=west] at (105.3,0) {$+\,\cdots$};
		\end{scope}
	\end{tikzpicture}%
}

\definecolor{ebdGibbsCompact}{HTML}{FFE9BB}

\newcommand{\twoSidedEBD}{%
  \begin{tikzpicture}[
    x=1cm,
    baseline=(current bounding box.center),
    line width=.65pt,font=\normalsize,
    every node/.style={inner sep=1pt},
    decoration={zigzag,segment length=2.2pt,amplitude=.45pt}]
    \foreach \x in {-1.5,-1.2,-.9,-.6,-.3,0,.3,.6,.9,1.2,1.5}
      \draw (\x,-.425)--(\x,.425);
    \filldraw[fill=ebdGibbsCompact]
      (-1.7,-.225) decorate{--(-1.7,.225)}--(1.7,.225)
      decorate{--(1.7,-.225)}--cycle;
    \node at (0,0) {$e^{-\beta H}$};
    \node at (2.025,0) {$=$};
    \begin{scope}[shift={(4.05,0)}]
      \foreach \x in {-1.5,-1.2,-.9,-.6,-.3,0,.3,.6,.9,1.2,1.5}
        \draw (\x,-1)--(\x,1);
      \filldraw[fill=ebdGibbsCompact]
        (-1.7,-.225) decorate{--(-1.7,.225)}
        --(-.15,.225)--(-.15,-.225)--cycle;
      \filldraw[fill=ebdGibbsCompact]
        (.15,-.225)--(.15,.225)--(1.7,.225)
        decorate{--(1.7,-.225)}--cycle;
      \node at (-.925,0) {$e^{-\beta H_{\setminus k}}$};
      \node at (.925,0) {$e^{-\beta H_{\setminus k}}$};
      \filldraw[fill=X]
        (-1.7,-.825) decorate{--(-1.7,-.375)}--(1.7,-.375)
        decorate{--(1.7,-.825)}--cycle;
      \node at (0,-.6) {$X$};
      \draw[fill=tensor] (-.8,.375) rectangle (.8,.825);
      \node at (0,.6) {$M$};
      \draw (0,1.05) node[anchor=south] {\footnotesize $k$};
    \end{scope}
  \end{tikzpicture}%
}

\begin{document}
\preprint{MIT-CTP/6113}
\title{One-dimensional quantum Gibbs states in constant circuit depth}
\author{Sa\'ul Pilatowsky-Cameo}
\affiliation{Center for Theoretical Physics --- a Leinweber Institute, Massachusetts Institute of Technology, Cambridge, MA 02139, USA}
\author{Georgios Styliaris}
\affiliation{Max Planck Institute of Quantum Optics, Hans-Kopfermann-Str 1, Garching 85748, Germany}
\author{Ainesh Bakshi}
\affiliation{Department of Computer Science, New York University, New York, NY 10012, USA}
\author{Daniel Malz}
\affiliation{Department of Physics, University of Basel, Switzerland}

\begin{abstract}
We prove that the Gibbs state of any spin chain can be exactly decomposed into a mixture of product states over contiguous blocks of qubits with bounded length, at any constant temperature. As a consequence, any such Gibbs state can be prepared by local unitary circuits of constant depth, with efficient classical preprocessing. We place strong limits on the spatial structure of thermal entanglement: both the entanglement depth and entanglement width (measures of multiparticle entanglement) are bounded by constants, while the localizable entanglement (the bipartite entanglement that can be generated with the aid of local measurements and classical communication) vanishes exactly beyond a constant threshold distance. Our result also has implications for quantum thermalization: for any translation-invariant spin-chain Hamiltonian with nondegenerate spectral gaps, typical low-complexity pure states drawn from a maximally entropic ensemble at any constant effective temperature become locally indistinguishable from the Gibbs state upon unitary evolution. Additionally, we obtain an analogous decomposition for fermionic systems: the  Gibbs state of any 1D local fermionic Hamiltonian is exactly a mixture of constant-depth local fermionic circuits acting on Gaussian states.
\end{abstract}
\maketitle

Quantum many-body systems can host large amounts of entanglement and complexity, leading to rich phenomena without classical analogues. At thermal equilibrium, however, these systems are described by universal states: Gibbs states determined by only a few macroscopic parameters, such as temperature. Understanding how many of the complex, distinctly quantum features can still manifest at thermal equilibrium is one of the central questions in quantum statistical physics~\cite{Alhambra2023}. This question can be made precise in several  complementary ways. In terms of the physical structure of the Gibbs state, one can ask for the presence of quantum correlations, as quantified by a measure of entanglement~\cite{Horodecki2009}. At the same time, from a computational perspective, one can ask for the minimum amount of quantum resources required to prepare the state on a quantum device, quantified for instance by the circuit depth complexity.

In one dimension, Gibbs states have exponentially decaying correlations at any positive
temperature, both for local observables~\cite{Araki1969,PerezGarciaPerezHernandez2023,KimuraKuwahara2025,Bergamaschi2026}
and mutual information~\cite{Bluhm2022,Scalet2021,Kato2019,KuwaharaKatoBrandao2020}, which satisfies an area law~\cite{Wolf2008,Kuwahara2021}. However, quantifying the purely quantum correlations is more delicate. Early studies of specific models, through exact solutions, entanglement witnesses, and numerics, suggested that thermal entanglement is short ranged and vanishes above a finite temperature~\cite{Guhne2005,Toth2005,Markham2008,Ferraro2008,Cavalcanti2008,Sherman2016}, and rigorous results valid for arbitrary local Hamiltonians have  appeared recently~\cite{Berta2018,Brandao2019,Lu2020,LuHsieh2020,Wu2020,Kuwahara2021,Kuwahara2022,kuwahara2025clustering,Bakshi2024,Bakshi2026,Scalet2026,Putterman2026}. Specifically, the total amount of bipartite entanglement, as quantified by the entanglement of formation~\cite{Bennett1996a} and the Schmidt number~\cite{Terhal2000}, has been shown to remain bounded independently of system size at any finite temperature~\cite{Kuwahara2021,Bakshi2026}. Furthermore, after tracing out an intermediate region, the entanglement of formation between separated regions decays exponentially with their distance~\cite{kuwahara2025clustering}, and was recently shown to vanish altogether beyond a sufficiently large separation~\cite{Scalet2026}.

On the computational side, thermal states admit efficient tensor-network descriptions~\cite{Verstraete2004mpdo,ZwolakVidal2004,Hastings2006,Pirvu2010,Kliesch2014,Molnar2015,Berta2018,GuthJarkovsky2020,Kuwahara2021,Huang2021,AlhambraCirac2021,Bakshi2026}, which underlie early preparation schemes~\cite{BilginBoixo2010}. A long line of work on the convergence of dissipative dynamics~\cite{KastoryanoBrandao2016,Bardet2021,Capel2020,Bardet2023,Bardet2024,Kochanowski2025} has recently culminated in quantum Gibbs samplers that provably prepare thermal states of general local Hamiltonians efficiently~\cite{Chen2023,ChenKastoryano2025,Ding2025,RouzeFranca2026a,RouzeFranca2026b,ChenRouze2025,Bakshi2025,Hahn2026,bakshi2026rapid,Bergamaschi2026b}, although the best bounds in one dimension can currently only guarantee a mixing time linear in system size \cite{Bergamaschi2026}. Reference~\cite{Bergamaschi2026} additionally presents an adiabatic preparation algorithm,  which has the best known circuit depth in 1D, scaling only polylogarithmically with the system size, at any nonzero temperature.

In this work, we show that 1D Gibbs states at any constant positive temperature can be prepared by an ensemble of quantum circuits whose depth is independent of system size. Our main structural result expresses the Gibbs state exactly as a classical mixture of pure states that factorize over contiguous blocks of bounded size, a property we call {\it block separability}. This decomposition explicitly separates the quantum operations needed to prepare the states from the classical computation involved in sampling the ensemble, a distinction important for practical purposes (classical computation may be significantly cheaper), but also from a fundamental perspective, as it allows us to explicitly bound the amount of entanglement present in the Gibbs state.

Our result implies strong bounds on multipartite entanglement, quantified by the entanglement depth~\cite{SoerensenMolmer2001} and width~\cite{WolkGuhne2016}, and also a sudden death of long-range localizable entanglement~\cite{Verstraete2004,Popp2005} in 1D. Beyond a finite length scale, no entanglement can be localized between two regions using local measurements on the remaining spins and classical communication. 

We also establish an analogous result for fermionic systems. At any fixed positive temperature, the Gibbs state is a mixture of pure Gaussian states transformed by constant-depth fermionic circuits acting on disjoint contiguous blocks of bounded size, which can be classically sampled efficiently, generalizing the results of Ref.~\cite{Ramkumar2026} to arbitrary temperatures in 1D. In what follows, we present these results in detail.

\textit{Block separability}.---Consider the Gibbs state of any geometrically local Hamiltonian $H = \sum_{j=1}^{N}h_j$ in a one-dimensional ring of qubits, at any inverse temperature $\beta$,
\begin{equation}
	g_\beta = \frac{\exp(-\beta H)}{\tr\exp(-\beta H)}.
	\label{eq:Gibbs-state}
\end{equation}  Here, each local term $h_j$ acts on $r\geq 2$ adjacent qubits, i.e., $\supp(h_j)=[j,j+r-1]$ \footnote{We use the notation $[a,b]\coloneqq \{a,a+1,\dots,b\}$.} and has bounded operator norm $\norm{h_j}_\infty\leq1.$ The sum $j+r-1$ is understood modulo $N$, allowing for periodic boundary conditions. Our main result is the following.
\begin{theorem}[Block separability of the Gibbs state at any temperature in 1D]
\label{th:01}
For any constant inverse temperature $0\leq\beta<\infty$, the Gibbs
state can be decomposed exactly as a mixture of product states
\begin{equation}
\label{th:gibbsstateblcoks}
    g_\beta=\sum_\mu p_\mu \bigotimes_{s=1}^q\dyad{\psi^{(\mu)}_{s}}_{A_s^{(\mu)}},
\end{equation}
over partitions of the ring into contiguous blocks
$A_1^{(\mu)},A_2^{(\mu)},\ldots,A_q^{(\mu)}$, each containing at most a constant number of qubits $C_{\beta}\coloneqq 
\exp(\exp(c_r\widetilde{\beta}))$, where  $\widetilde\beta=\max(1,\beta)$ and $c_r=\exp(16r)$.
Furthermore, each tensor factor $\ket{\psi_{s}^{(\mu)}}$ is a matrix product state (MPS)~\cite{PerezGarcia2007,Cirac2021} of bond dimension $\chi\leq C_{\beta}$. 
\end{theorem}
\noindent Theorem~\ref{th:01} generalizes the exact high-temperature
separability results of Ref.~\cite{Bakshi2024} to arbitrary temperatures in 1D,
where separability is understood in a coarse-grained sense, after blocking a
constant number of qubits in the chain. Additionally, we show that the ensemble of Theorem~\ref{th:01} can be sampled efficiently.
\begin{theorem}[Classical sampling algorithm]
    There exists a classical algorithm with run time $\poly(N,1/\varepsilon)$ that outputs  a sequence of contiguous blocks $A_s^{(\mu)}$ and MPSs $\ket{\tilde\psi_s^{(\mu)}}$ with the properties described in Theorem~\ref{th:01}, such that
    \begin{equation}
        \bigg\|\mathbb{E}\Big[\bigotimes_{s=1}^{q} \dyad{\tilde\psi_s^{(\mu)}}_{A_s^{(\mu)}}\Big] - g_\beta\bigg\|_1< \varepsilon.
    \end{equation}
    \label{th:02}
\end{theorem}
\noindent Theorems~\ref{th:01} and~\ref{th:02} have fundamental consequences on the depth complexity and entanglement of 1D Gibbs states, and additionally imply a rigorous result on quantum thermalization under unitary dynamics. Next, we explain these corollaries, and afterwards present a proof sketch for the two theorems. 

\textit{Constant depth complexity}.---Since each MPS $\ket{\psi_s^{(\mu)}}$ in Eq.~\eqref{th:gibbsstateblcoks} can be prepared by a circuit of constant depth $O(C_{\beta}\chi^2)=\exp(\exp(O(\widetilde\beta)))$ \cite{Schon2005}, Theorems \ref{th:01} and \ref{th:02} immediately imply the following.
\begin{corollary}[1D Gibbs states in constant depth]
	The Gibbs state is a mixture of states preparable by  unitary circuits  $U_\mu$ consisting of nearest-neighbor gates, with depth $\exp(\exp(O(\widetilde\beta)))$
	\begin{equation}
		g_\beta = \sum_{\mu}p_\mu U_\mu(\proj0)^{\otimes N}U_\mu\dagg.
		\label{eq:exact-preparation}
	\end{equation}
	 Furthermore, a description of the circuits $U_\mu$ can be classically sampled in time $\poly(N,1/\varepsilon)$, incurring an error $\varepsilon$ in the trace norm.
    \label{co:gibbs-prep}
\end{corollary}
\noindent Corollary~\ref{co:gibbs-prep} provides a quantum algorithm for thermal state preparation in 1D whose circuit depth complexity is a constant, independent of $N$, at arbitrary constant temperatures. In  all previous algorithms for Gibbs state preparation in this regime, the circuit depth diverges with the system size~\cite{BilginBoixo2010,Brandao2019,ChenKastoryano2025,Ding2025,RouzeFranca2026a,RouzeFranca2026b,ChenRouze2025,Hahn2026,Bergamaschi2026} or requires quasilocal gates~\cite{Kato2019}.

\textit{Strong bounds on entanglement}.---Theorem~\ref{th:01} implies bounds on the amount and range of multipartite entanglement in the Gibbs state. Multipartite entanglement is usually quantified by the {\it entanglement depth}~\cite{SoerensenMolmer2001,Guhne2005} defined as the smallest integer $k$ such that a state is a mixture of block-product states on $k$ qubits. This quantity prominently appears in the field of quantum sensing, as it bounds the quantum Fisher
information~\cite{Pezze2009,Hyllus2012,Toth2012,Pezze2018}, which has been extensively studied in thermal states~\cite{Hauke2016,Gabbrielli2018,FrerotRoscilde2018,FrerotRoscilde2019}. The entanglement width~\cite{WolkGuhne2016} is defined similarly, but $k$ bounds the spatial extent of the blocks rather than  the number of qubits they contain. Theorem~\ref{th:01} directly places a bound on both of these quantities.
\begin{corollary}[Finite multipartite entanglement]
\label{cor:entanglement-bounds}
The Gibbs state \(g_\beta\) has entanglement depth and entanglement width bounded by the constant $C_{\beta}$. 
\end{corollary}

Theorem~\ref{th:01} also implies a spatial sudden death of {\it localizable entanglement}~\cite{Verstraete2004,Popp2005}, which quantifies how much entanglement can be produced between two chosen regions by performing local measurements assisted by classical communication on the rest of the system. Specifically, 
the \textit{multi-user entanglement of assistance}~\cite{DiVincenzo1998,Colomer2023} of $\rho_{ABC}$, where $B=\bigcup_i B_i$ and each $B_i$ is a single site,  is defined by maximizing the average entanglement between $A$ and $C$ after local operations and classical communication (LOCC) on each site $B_i$,
\begin{equation}
    L^E_{A,C}(\rho_{ABC}):=\sup_{\mathcal P}\sum_j p_j E(\rho^{(j)}_{AC}).
\end{equation}
Here, we choose $A$, $B$, $C$ to be regions of the 1D chain,  $\mathcal P=\{P_B^{(j)}\}_j$ is a LOCC positive operator-valued measure (POVM) on $B$, with normalized post-measurement states  $\rho_{AC}^{(j)}\propto\tr_B((\1_A\otimes P_B^{(j)}\otimes\1_C)\rho_{ABC})$,
and \(E\) is any faithful bipartite measure of entanglement (e.g., the entanglement of formation~\cite{Bennett1996a}).

\begin{corollary}[Spatial sudden death of localizable entanglement]
\label{cor:entanglement-bounds-loc-ent}
If the distance between regions $A$ and $C$ is larger than the constant $C_\beta$, no entanglement can be localized from the Gibbs state through LOCC on $B$
\begin{equation}
    L^E_{A,C}(g_\beta)=0.
\end{equation}
\end{corollary}

Corollary~\ref{cor:entanglement-bounds-loc-ent} strengthens the main result of Ref.~\cite{Scalet2026}, which states that after tracing out $B$, the reduced Gibbs state becomes separable for sufficient separation. Tracing out may
destroy entanglement that could still be localized via measurement. For example, an \(N\)-qubit GHZ state has separable reduced states after tracing out any qubit, but it has localizable entanglement for arbitrarily distant regions, and in fact also entanglement depth and width growing as
\(N\). Corollaries~\ref{cor:entanglement-bounds} and~\ref{cor:entanglement-bounds-loc-ent} guarantee that this is not the case for the Gibbs state in 1D. Our decomposition also implies the vanishing, beyond the same finite length scale, of convex-roof measures of long-range entanglement such as the co(QCMI) introduced in Ref.~\cite{Grover2025}. We remark that, in contrast to 1D, it was recently shown that the Gibbs state of a certain 2D commuting Hamiltonian does support long-distance localizable entanglement at finite temperature~\cite{Harley2026}.
\begin{figure*}[tb]
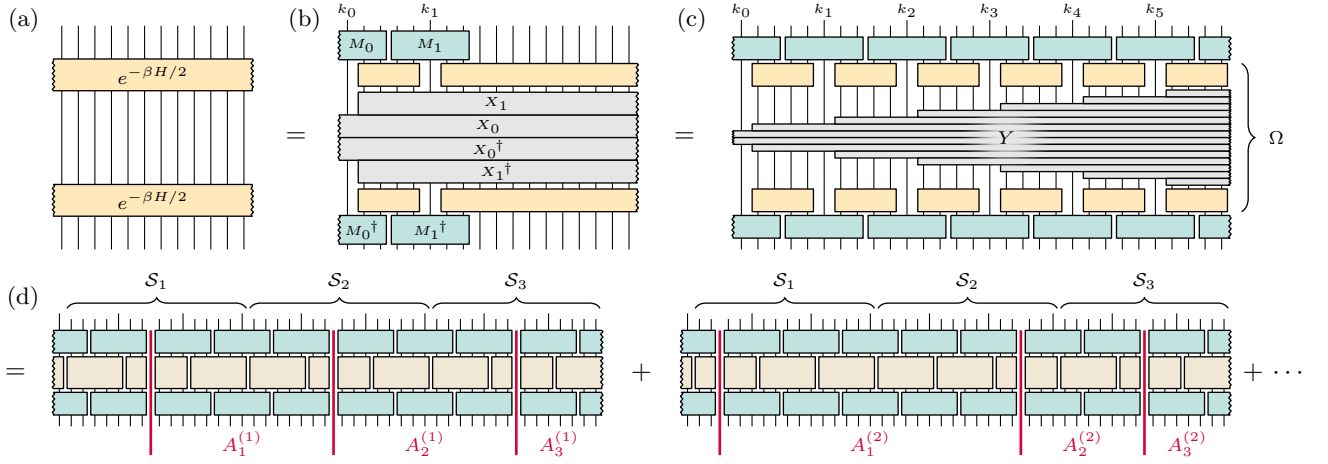

	\constructionFigure
	\caption{
		Schematic of the construction. (a)~Unnormalized Gibbs state $e^{-\beta H}$ on a ring (vertical wires are qubits). (b)~Applying the entanglement bulk decomposition at sites $k_0$ and $k_1$ yields local operators $M_j$ (teal), Gibbs states of the remaining segments (yellow), and QPIs $X_j$ (gray). (c)~Iterating over all $k_j$ leaves, between the $M_j$, an operator $\Omega$ containing the separable staircase $Y$. (d)~$\Omega$ is a sum of positive terms, each a product over the blocks $[k_{j-1}+1,k_j]$ with a cut (red) in every superblock $\mathcal S_s$. Since no $M_j$ crosses a cut, each term factorizes over the regions $A_s^{(\mu)}$ of Theorem~\ref{th:01}.
	}
	\label{fig:construction}
\end{figure*}

\textit{Quantum thermalization}.---Our results also have consequences for quantum
dynamics: they allow us to prove that quantum thermalization~\cite{Gogolin2016,Alhambra2023} must occur in generic translation-invariant 1D Hamiltonians, extending the high-temperature results of Ref.~\cite{PilatowskyCameo2026} to all positive temperatures in one dimension. We briefly introduce the basic ideas here.

An initial state $\ket{\psi}$ evolving unitarily $\ket{\psi(t)}=e^{-iHt}\ket{\psi}$ is said to thermalize when it becomes locally indistinguishable from the Gibbs state. To make a precise probabilistic statement about this phenomenon, we consider an ensemble
$\mathcal{E} =
    \left\{
        \left(p(\psi),\ket{\psi}\right)
    \right\}_\psi$ 
of initial states and say that it thermalizes if for a local region $A$, the reduced states $\psi_A(t)
    \coloneqq
    \tr_{A^c}\!\left(\dyad{\psi(t)}\right),
$ and $
    g_{\beta,A}
    \coloneqq
    \tr_{A^c}(g_\beta)$
 are close in trace norm,
\begin{equation}
    \label{eq:thermalization}
    \mathbb{E}_{\ket{\psi}\sim\mathcal{E}}
    \mathbb{E}_t
    \left[
        \norm{\psi_A(t)-g_{\beta,A}}_1
    \right]
    \leq
    \varepsilon(N),
\end{equation}
where \(\varepsilon(N)\) captures finite-size fluctuations and vanishes
in the thermodynamic limit. The first average is over the ensemble $\mathcal{E}$; the second is the infinite-time average, defined by $\mathbb{E}_t[(\,\cdot\,)]
    \coloneqq
    \lim_{T\to\infty}
    \frac{1}{T}
    \int_0^T \mathrm{d}t\, (\,\cdot\,)~$\footnote{This limit exists by the Kronecker--Weyl theorem, since finite-dimensional quantum evolution is quasiperiodic.}, so Eq.~\eqref{eq:thermalization} implies that, for every 
observable $\norm{O_A}_\infty\leq 1$ supported on \(A\), $\bra{\psi(t)}O_A\ket{\psi(t)} \approx\tr(g_{\beta}O_A) \pm \varepsilon(N)$ for typical initial states $\ket{\psi}\sim\mathcal{E}$ and late times $t$.

Quantum thermalization is nontrivial when the initial state $\ket{\psi}$ is far from equilibrium and has low complexity, as then entanglement must grow through the unitary dynamics and saturate to an extensive value, such that the entanglement entropy matches the thermal entropy of the Gibbs state. We obtain a statement about finite-temperature quantum thermalization by considering ensembles of initial states $\mathcal{E}_{\beta,\mathcal{C}}$ with the following three properties~\cite{PilatowskyCameo2026}: 
    \\(P1)~\textit{Low complexity:} For all $\ket{\psi}\in \mathcal{E}_{\beta,\mathcal{C}}$, $\ket{\psi}$ can be prepared by a circuit of depth $\mathcal{C}(N)$, which grows at most subpolynomially in $N$, 
    \\(P2)~\textit{Fixed temperature:} $\mathbb{E}_{\ket{\psi}\sim\mathcal{E}_{\beta,\mathcal{C}}}[\bra{\psi}{H}\ket{\psi}]=\tr(g_\beta H)$,
    \\(P3)~\textit{Maximally entropic ensemble:} The average state $\rho_{\mathcal{E}}=\mathbb{E}_{\ket{\psi}\sim\mathcal{E}_{\beta,\mathcal{C}}}[\dyad{\psi}]$ achieves the maximal von Neumann entropy $-\tr(\rho \log(\rho))$ among all ensembles satisfying (P1) and (P2) at the same inverse temperature $\beta$ and complexity $\mathcal{C}$. This property says that the ensemble is large, in the specified sense.

 An example of such an ensemble is already given by Theorem~\ref{th:01}, since the average state is the maximally entropic Gibbs state $\rho_{\mathcal{E}}=g_\beta$, and the complexity is constant $\mathcal{C}(N)=\mathcal{C}_\beta=\exp\exp O(\widetilde\beta)$. By the results of Ref.~\cite{PilatowskyCameo2026}, the existence of this ensemble implies the thermalization of any ensemble satisfying (P1--3), provided that \(H\) is
translation invariant, meaning that the local terms \(h_j\) are
translates of one another around the ring, and that \(H\) has
nondegenerate spectral gaps~\cite{Tasaki1998,Reimann2008,Linden2009}. The latter condition means that, for any
energy eigenvalues \(E_a,E_b,E_c,E_d\), with \(a\neq b\) and
\(c\neq d\),
\begin{equation}
    \label{eq:nondegengaps}
    E_a-E_b=E_c-E_d
    \quad\Longrightarrow\quad
    a=c
    \quad\text{and}\quad
    b=d.
\end{equation}

\begin{corollary}[Quantum thermalization at any temperature
in one dimension]
\label{th:thermalization}
Let \(H\) be translation invariant and have nondegenerate spectral gaps
in the sense of Eq.~\eqref{eq:nondegengaps}. Then, for any 
$0\leq \beta<\infty$ and  $\mathcal C(N)\geq \mathcal{C}_\beta$, any ensemble of initial states $\mathcal{E}_{\beta,\mathcal{C}}$ satisfying (P1--3) thermalizes in the sense of
Eq.~\eqref{eq:thermalization}. For every region \(A\) of constant
size, the finite-size fluctuations satisfy
$\varepsilon(N)=O\!\left(N^{-1/2+\delta}\right)$ for any $\delta>0$. 
\end{corollary}
\noindent In contrast to prior works~\cite{Muller2015,Farrelly2017}, \cref{th:thermalization} does not require assumptions on the energy distribution of the initial state.  See App.~\ref{app:therm-proof} for a proof. We now turn to explain the main ideas behind the proof of Theorem~\ref{th:01}, with a complete proof in App.~\ref{app:proof-of-th-01}.

\textit{Proof sketch of Theorem~\ref{th:01}}.---Our first ingredient is the following \textit{entanglement bulk decomposition (EBD)}~\cite{Bakshi2026}, which expresses the Gibbs state in terms of the Hamiltonian $H_{\setminus k}$ obtained by removing all interactions acting on site $k$: $e^{-\beta H}=Me^{-\beta H_{\setminus k}}X$, or diagrammatically,
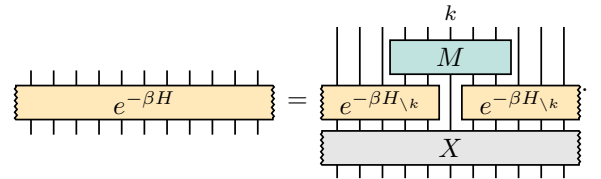
\begin{equation*}
    \twoSidedEBD.
\end{equation*}
Here, $M$ is a local operator acting on $m=\exp(O(\widetilde\beta))$ sites around $k$, and $X$ is a {\it quasilocal perturbation of the identity} (QPI): $X=\1+\sum_jF_j$, where $F_j$ is supported on $j$ qubits around $k$ and $\norm{F_j}_\infty\leq \gamma^{j}$, for a decay constant $\gamma<1$ that can be chosen freely at the expense of increasing $m$. See App.~\ref{app:proof-of-th-01} for the precise statement and proof, based on Ref.~\cite{Bakshi2026}.

We apply this decomposition at inverse temperature $\beta/2$ to a sequence of sites $k_0,k_1,\dots$ spaced by $m$ around the ring, as shown in Fig.~\ref{fig:construction}~(a)-(c). These sites define blocks $B_j=[k_{j-1}+1,k_j-1]$ and we obtain operators $M_j$ with disjoint supports connecting neighboring blocks, such that $g_\beta\propto (\prod_jM_j)\,\Omega\, (\prod_jM_j)^\dagger$, where
\begin{equation}
    \Omega=\prod_{j=1}^n e^{-\beta H_{B_j}/2}\,Y\,\prod_{j=1}^n e^{-\beta H_{B_j}/2}.
\end{equation}
In App.~\ref{app:periodic-staircase}, adapting arguments from Refs.~\cite{Bakshi2024,Bakshi2026}, we show that the middle operator $Y$, which is a staircase of QPIs [gray boxes in Fig.~\ref{fig:construction}(c)], is separable:
\begin{equation}
\label{eq:Yopdefinition}
    Y=X_{n-1}\cdots X_0X_{0}^\dagger\cdots X_{n-1}^\dagger=\smash{\sum_\nu} w_\nu \,\smash{\bigotimes_j} \,\omega_j^{(\nu)},
\end{equation}
where each $\omega_j^{(\nu)}$ is supported on $B_j\cup\{k_j\}=[k_{j-1}+1,k_j]$. Thus, we can write $\Omega=\sum_\nu w_\nu\bigotimes_jG_j^{}$, with $G_j=e^{-\beta H_{B_j}/2}\omega_j^{(\nu)}e^{-\beta H_{B_j}/2}$. 

 We separate each term  as $G_j=G_j'+\Delta_j$, where both summands are positive semidefinite but crucially $G_j'$ factorizes across the midpoint of the block $B_j$ and  $\Delta_j\preceq(1-\eta)G_j$ for some small constant $\eta=\exp(-O(m))>0$. Grouping $L$ consecutive blocks into a {\it superblock}  $\mathcal S$, we expand $\bigotimes_{j\in\mathcal S}(G_j'+\Delta_j)$. Every summand containing a term $G_j'$ has a cut at the midpoint of that block. The only exception is the term $\bigotimes_{j\in\mathcal S}\Delta_j$, whose size relative to $\bigotimes_{j\in\mathcal S}G_j$ is suppressed by a factor $(1-\eta)^L$, which allows us to prove  separability. Specifically, we choose a block $j_0\in\mathcal S$ and collect all summands containing $G_{j_0}'$ into $R=G_{j_0}'\otimes\bigotimes_{j\in\mathcal S\setminus\{j_0\}}G_j$. For sufficiently large $L=\exp(O(m))$, the operator $R+\bigotimes_{j\in\mathcal S}\Delta_j$ is separable across the midpoint of block $j_0$.

Repeating this construction in every superblock of $mL=\exp(\exp(O(\widetilde\beta)))$ qubits gives a positive decomposition of $\Omega$ with at least one cut per superblock [Fig.~\ref{fig:construction}~(d)]. Since none of the operators $M_j$ crosses these cuts, the operators between the cuts remain unentangled upon conjugation by the $M_j$, and after spectrally decomposing each block and applying the operators $M_j$, we obtain the MPSs in Theorem~\ref{th:01}. The blocks $A_s^{(\mu)}$ are defined between the cuts of each superblock, and thus contain at most $ 2mL\leq C_{\beta}$ qubits.

\textit{Sampling algorithm}.---The algorithm of Theorem~\ref{th:02} follows the construction outlined above step by step. It proceeds in three stages  (see App.~\ref{app:alg} for details). First, a term $\nu$ can be sampled from the separable decomposition of the staircase $Y$ of Eq.~\eqref{eq:Yopdefinition} in $\poly(N,1/\varepsilon)$ time, by adapting the algorithm of Ref.~\cite{Bakshi2026}. Second, once $\nu$ is fixed, one must decide which summand of the expansion of $\bigotimes_{j\in\mathcal S}(G_j'+\Delta_j)$ is sampled in each superblock, and hence where the resulting cut lies. Since none of the $M_j$ crosses a cut, the weight of any global configuration factorizes across the cuts, with each factor depending only on the configuration of the two superblocks containing neighboring cuts. This is a classical nearest-neighbor distribution on a ring which can be sampled exactly by a transfer-matrix (forward-backward) sweep in time $O(N)$. Finally, once all cuts are fixed, the remaining operator is a tensor product of positive operators on the blocks, and drawing its eigenvectors with probability proportional to their weight yields a product of pure states, which becomes the product of MPSs of Theorem~\ref{th:01} after applying the operators $M_j$.

\textit{Fermionic systems}.---We obtain an analogous result for fermionic Hamiltonians. Consider a ring of $N$ fermionic modes with annihilation operators satisfying $\{c_i,c_j^\dagger\}=\delta_{ij}$ and $\{c_i,c_j\}=0$. Let $H=\sum_j h_j$, where each $h_j$ is a linear combination
of products containing an even number of creation and annihilation
operators on sites $j,\ldots,j+r-1$,
with $\norm{h_j}_\infty\leq1$.
\begin{theorem}[Fermionic Gibbs states in constant depth]
\label{th:03}
At any  inverse temperature $0\leq\beta<\infty$, the fermionic Gibbs state admits an exact decomposition
\begin{equation}
    g_\beta
    =\sum_\mu p_\mu U_\mu
    \ket{G_\mu}\!\bra{G_\mu}U_\mu^\dagger,
\end{equation}
where each $\ket{G_\mu}$ is a fermionic Gaussian state and $U_\mu=\bigotimes_{s=1}^{q}U_{\mu s}$, with each $U_{\mu s}$ a circuit of nearest-neighbor parity-preserving fermionic gates of depth $\exp(\exp(O(\widetilde\beta)))$ supported on a block $A_s^{(\mu)}$ which contains at most $C_{\beta}$ contiguous modes ($C_{\beta}$ defined as in Theorem~\ref{th:01}). Furthermore, classical descriptions of $\ket{G_\mu}$ and $U_\mu$ can be
sampled in time $\poly(N,1/\varepsilon)$,  incurring an error
$\varepsilon$ in the trace norm.
\end{theorem}
At sufficiently high temperature, Ref.~\cite{Ramkumar2026} proved that the fermionic Gibbs state is exactly a mixture of Gaussian states. Theorem~\ref{th:03} shows that, in 1D, one only needs to apply a layer of constant-depth circuits to reach arbitrary finite $\beta$. Theorem~\ref{th:03} follows from Theorem~\ref{th:01} by performing a Jordan--Wigner transformation. See App.~\ref{app:ferm} for details.

\textit{Discussion and outlook}.---We have provided a decomposition of one-dimensional Gibbs states into mixtures of states that can be prepared by quantum circuits of constant depth. 
While our circuit depth is a constant in the system size, our bound scales doubly exponentially with the inverse temperature $\beta$, and an immediate question is whether this dependence can be improved. We note that the approximate, polylogarithmic-depth preparation algorithm of Ref.~\cite{Bergamaschi2026} also has a $\beta$ dependence which is at least doubly exponential. However, we place a universal bound, so the dependence may be improved for specific models. For a given model, the circuit depth required can be obtained numerically by iteratively increasing the blocking size $m$ until the tails in the Araki expansionals are sufficiently small~\cite{Bakshi2026}. Such an investigation would allow one to assess whether our dependence on temperature is optimal. Specifically, do there exist families of Hamiltonians for which we can find a lower bound for the required circuit depth that matches $\exp\exp O(\beta)$?

While our decomposition is exact, in practice only approximate preparation is required. 
Allowing for approximations may improve the temperature dependence. Reference~\cite{Kuwahara2021} finds a decomposition into matrix product states whose bond dimension scales subexponentially in $\beta$. 
When compiled into a circuit using standard sequential preparation of MPS~\cite{Schon2005}, this yields circuits with better scaling in $\beta$ compared to the doubly exponential scaling in Ref.~\cite{Bakshi2026} and the present work, but at the expense of quasi-linear scaling in $N$.

Generalizing our results beyond the setting studied here is also an important direction.
Constant-depth preparation of Gibbs states of qudit Hamiltonians should follow from a straightforward adaptation of the present argument, and it would also be interesting to analyze bosonic systems, with infinite local Hilbert-space dimensions. Beyond one-dimensional systems, we ask whether Gibbs states in two dimensions are also mixtures of short-range entangled states. Such a result would prove the absence of finite-temperature topological order (FTTO) in two dimensions~\cite{Hastings2011}. 
It is known that FTTO is absent in one dimension~\cite{Kato2019} (which is substantially tightened by our result~\footnote{Reference~\cite{Kato2019} provides a mixed circuit, made up of channels with support on polylogarithmically
many sites, while Corollary~\ref{co:gibbs-prep} gives an exact mixture of pure states prepared by
constant-depth unitary circuits made up of strictly local gates.}), and that FTTO is possible in three dimensions~\cite{Zhou2025a}.
Hastings proved the absence of FTTO in two dimensions for commuting Hamiltonians~\cite{Hastings2011}, but the general case is still open.

\section*{Acknowledgments}
We thank S. Choi, I. Cirac, and R. Trivedi for insightful conversations. We acknowledge the use of GPT 5.5-6 and Claude Fable 5-5.1 for technical assistance in our proofs. This work was supported by the U.S. Department of Energy, Office of Science, Award Number DE-SC0021013.

\onecolumngrid

\appendix
\newcounter{appthm}
\numberwithin{appthm}{section}
\renewcommand{\theappthm}{\thesection\arabic{appthm}}
\makeatletter
\newcommand{\appthmify}[1]{
  \numberwithin{#1}{section}
  \expandafter\def\csname the#1\endcsname{\thesection\arabic{#1}}
}
\makeatother
\appthmify{theorem}
\appthmify{corollary}
\appthmify{lemma}
\appthmify{proposition}
\appthmify{definition}
\clearpage
\section{Proof of Theorem~\ref{th:01}}
\label{app:proof-of-th-01}
Here we prove Theorem~\ref{th:01}. The proof has three ingredients. First, we define quasilocal perturbations of the identity (QPIs) and state that a staircase of QPIs is separable (Lemma~\ref{lm:staircase-separability}, proved in App.~\ref{app:periodic-staircase}). Second, we extend the entanglement bulk decomposition of Ref.~\cite{Bakshi2026} to a two-sided version that allows any qubit of a ring to be removed (Theorem~\ref{th:tsEBD}). Third, we iterate this decomposition to split the ring into blocks and superblocks, and show that within every superblock the resulting operator admits a positive decomposition with a product cut across the midpoint of some block (Proposition~\ref{prop:omegaop}).
\subsection{Quasilocal perturbations of the identity and separability}

We begin with the definition of the operators that will encode the classical correlations of the Gibbs state.
\begin{definition}[Quasilocal perturbation of the identity (QPI)]

		An operator $X$ is a \emph{QPI with decay $\gamma$ anchored at $k$} if it can be written as
		\begin{equation}
			X = \1 + \sum_{j=1}^N F_j,\qquad\norm{F_j}_\infty\leq\gamma^j,\qquad\supp(F_j)\subseteq[k-\floor{j/2},k+\ceil{j/2}-1].
			\label{eq:two-sided-X}
		\end{equation}
    When \(X\) acts on a ring, the support condition is interpreted cyclically:
\([k-\lfloor j/2\rfloor,k+\lceil j/2\rceil-1]\) denotes the cyclic interval of \(j\) sites, with site labels understood modulo \(N\) in \(\{1,\ldots,N\}\).
	\label{def:quasilocal-perturbation}
\end{definition}
\noindent This definition generalizes that of Ref.~\cite{Bakshi2026}, which considered only the case where $k$ is the first qubit. 
The crucial general property of QPIs is that they become separable if the decay is fast enough~\cite{Bakshi2024}. More generally, we show the following result:

\begin{lemma}[Staircases of QPIs are separable (variation of Theorem 6.3 in \cite{Bakshi2026})]
Let $1=k_0<k_1<\cdots<k_{n-1}<k_n=N+1$. Let   $X_j$ be QPIs with decay $1/112$ anchored at $k_j$ and supported on $[k_{j-1}+1,N]$ (for $X_0$ we allow periodic boundary conditions in the decay). Then
	\begin{equation}
		Y\coloneqq (X_{n-1}\cdots X_0)(X_{n-1}\cdots X_0)\dagg= \sum_\nu w_\nu\omega^{(\nu)}_1\otimes\cdots\otimes\omega^{(\nu)}_n
		\label{eq:Y-decomposition}
	\end{equation}
	with $w_\nu\geq 0$, and $\omega_j^{(\nu)}$ being product density matrices supported on $[k_{j-1}+1,k_j]$ with unit trace and $\omega_j^{(\nu)}\succeq 64^{k_{j-1}-k_j}\1$.
	\label{lm:staircase-separability}
\end{lemma}
\noindent See App.~\ref{app:periodic-staircase} for a proof, which closely follows the argument in Refs.~\cite{Bakshi2026,Bakshi2024}.

\subsection{Two-sided entanglement bulk decomposition}

The main ingredient in Ref.~\cite{Bakshi2026} is its Theorem 8.1, the \emph{entanglement bulk decomposition}, which allows one to write the Gibbs state on the full chain in terms of the Gibbs state on the chain with one qubit removed, dressed by a local operator $M$, carrying the quantum correlations, and a quasilocal perturbation of the identity $X$, carrying only classical correlations. In Ref.~\cite{Bakshi2026} the statement is made only for open chains, where the qubit being removed is the leftmost qubit.

Here, we generalize this result to allow any intermediate qubit to be removed, and for the Hamiltonian to possibly have periodic boundary conditions.

\begin{theorem}[Two-sided entanglement bulk decomposition]
	For any geometrically $r$-local Hamiltonian on a spin ring of length $N$ and any $0<\gamma<1$, the unnormalized Gibbs state at any inverse temperature $0<\beta<\infty$ can be decomposed as
	\begin{equation}
		e^{-\beta H}=Me^{-\beta H_{\setminus k}}X,
		\label{eq:EBD}
	\end{equation}
	where $M$ is supported on $m$ qubits around the $k$th qubit, i.e., $\supp(M)\subseteq [k-(m-1)/2,k+(m-1)/2]$, with $m=\exp[r\tilde\beta(20/\gamma)^{2r}]$ (rounded up to be an odd integer), $H_{\setminus k}$  is $H$ with all local terms acting on qubit $k$ removed, $X$ is a QPI with decay $\gamma$ anchored at $k$, and $\tilde\beta\coloneqq\max(1,\beta)$.

	\label{th:tsEBD}
\end{theorem}
\begin{proof}
This result can be derived from Theorem~8.1 in Ref.~\cite{Bakshi2026} directly, by folding the ring so that the qubit to be removed is at the left of an open chain. Specifically, order the sites as $(k, k-1, k+1, k-2, k+2, \dots)$. Applying this permutation to the Hamiltonian $H$ produces a new Hamiltonian $H^{(k)}$ that is also geometrically local, but with twice the range.
	Applying Theorem~8.1 of Ref.~\cite{Bakshi2026} to $H^{(k)}$ yields a decomposition $\exp(-\beta H^{(k)}) = M_{k}\exp[-\beta(H^{(k)}-H^{(k)}_1)] X,$
    where $H^{(k)}_1$ contains all the local terms acting on qubit $k$ (the leftmost qubit on the folded chain), and where $M_k$ is supported on the first $m=\exp[r\widetilde{\beta}(20/\gamma)^{2r}]$ sites in the folded chain, which corresponds to $\supp(M_k)=[k-(m-1)/2,k+(m-1)/2]$ in the original chain (rounding $m$ up to an odd integer), and $X$ is a QPI with decay $\gamma$ anchored at $1$ in the folded chain, i.e., at $k$ in the original chain. Note that we made the replacements $\mathfrak{K}\to 2r$ and $\beta\to r\beta$ in the constant $m$ as compared to Ref.~\cite{Bakshi2026}, to account for the doubling in locality, and to guarantee the normalization of local terms acting on the first qubit in the folded Hamiltonian, dividing $H^{(k)}$ by $r$ and absorbing the factor into $\beta$. We remark that this extra factor of $r$ could be removed by, instead of applying the result to the folded chain, re-doing the counting in Ref.~\cite{Bakshi2026} for two-sided growth sets. We retain the extra factor of $r$ here, since the argument is simpler. After unfolding the ring, $H^{(k)}-H^{(k)}_1$ becomes  $H_{\setminus k}$, as desired.\end{proof}

\subsection{Blocks and superblocks}

We  iterate Theorem~\ref{th:tsEBD} to obtain the following decomposition. 
\begin{corollary}
    For $1=k_0<k_1<\cdots<k_{n-1}<k_n=N+1 \equiv1$ and $0\leq \beta< \infty$, the Gibbs state $g_\beta$ can be written exactly as 
	\begin{equation*}
		g_\beta \propto \prod_{j=0}^{n-1}M_j\Omega \prod_{j=0}^{n-1}M_j^\dagger,
		\qquad \Omega = \prod_{j=1}^{n}e^{-\beta H_{B_j}/2} (X_{n-1}\cdots X_0)(X_{n-1}\cdots X_0)\dagg \prod_{j=1}^{n}e^{-\beta H_{B_j}/2},
	\end{equation*}
	where $\supp(M_j)=[k_j-(m-1)/2,k_j+(m-1)/2]$, with the support interval understood modulo $N$, $X_j$ are QPIs with decay $\gamma$ anchored at $k_j$ and supported on $[k_{j-1}+1,N]$ $(k_{-1}=0)$ and $H_{B_j}$ is a restriction of $H$, with all terms whose support is not contained in $B_j=[k_{j-1}+1,k_j-1]$ removed, and $m$ is defined as in Theorem~\ref{th:tsEBD}.
	\label{lm:BCP-mod}
\end{corollary}

We select $\gamma=1/112$, which guarantees that the staircase of QPIs is separable by Lemma~\ref{lm:staircase-separability}, and $m = \exp[r\tilde\beta(20\cdot 112)^{2r}]$, rounded up to an odd number. We select evenly spaced $k_j=1+jm$ for $j\in\{0,\ldots,n-1\}$, which splits the chain into $n=N/m$ blocks $[k_{j-1}+1,k_j]$ of length $m$, where the operators $M_0,\ldots,M_{n-1}$ act around the cuts between contiguous blocks. We will group consecutive blocks together into {\it superblocks} each containing $L=\exp(9m)$ elementary blocks, and hence $Lm$ qubits. Specifically, define a superblock $\mathcal{S}_s=\{(s-1)L+1,(s-1)L+2,\dots,sL\}$, containing all the block labels $j\in\mathcal{S}_s$, where $1\leq s\leq n/L$.  Here we have ignored rounding issues: if the chain does not divide exactly,  the remaining sites can be absorbed into an adjacent block/superblock without affecting the bounds. Let $\bar{k}_j\coloneqq \lfloor(k_{j-1}+k_j)/2\rfloor$ be the qubit before the midpoint  of the $j$th block. We show the following result.
\begin{proposition}[Middle operator has a cut in each superblock]
\label{prop:omegaop}
    The operator $\Omega$ of Corollary~\ref{lm:BCP-mod} can be written as a mixture $\Omega=\sum_\mu \Omega^{(\mu)}$, where each $\Omega^{(\mu)}$ is positive semidefinite and for each superblock $\mathcal{S}_s$ there exists $j_*(s,\mu)\in \mathcal{S}_s$ such that $\Omega^{(\mu)}$ is a product operator with a cut across the midpoint of the $j_*(s,\mu)$th block, that is,
    \begin{equation}
        \Omega^{(\mu)}=\bigotimes_{\smash{s=1}}^{\smash{n/L}} \Omega^{(\mu)}_{A_s^{(\mu)}}, \qquad\qquad A_s^{(\mu)}=[\bar{k}_{j_*(s,\mu)}+1,\bar{k}_{j_*(s+1,\mu)}],
\end{equation}
where each factor is positive \(\Omega_{A_s^{(\mu)}}^{(\mu)}\succeq0\) and is itself a product operator over each of the elementary blocks.
\end{proposition}
\noindent Proposition~\ref{prop:omegaop} together with Corollary~\ref{lm:BCP-mod} implies  Theorem~\ref{th:01}: since none of the operators $M_j$ crosses the midpoint of a selected block, the blocks $A_s^{\smash{(\mu)}}$ remain unentangled after conjugating $\Omega$ by the chain of $M_j$. Within each $A_s^{\smash{(\mu)}}$, we take a spectral decomposition of each of the elementary-block factors of $\Omega^{\smash{(\mu)}}_{A_s^{\smash{(\mu)}}}$. We refine the label $\mu$ to include the spectral choices and retain the same notation for the resulting pure-state ensemble. Since each $M_j$ acts on at most $\lceil m/2\rceil$ qubits of either
adjacent block, upon applying the $M_j$ chain the resulting pure
state is an MPS with bond dimension $\chi\leq 2^{m}
\leq C_\beta,$
as claimed. Moreover, the size of $A_s^{\smash{(\mu)}}$ is bounded by $2mL\leq C_\beta$. We now prove  Proposition~\ref{prop:omegaop}. We use the following general result.
\newcommand{\constlemm}[1]{\widetilde C_{#1}}
\begin{lemma}
   For any interval of the chain $\Lambda=[a,c]$ (possibly including the full ring), and $a<b<c$
    \begin{equation}
        \constlemm{\beta}^{-1}e^{-\beta H_{\Lambda\setminus b}} \preceq e^{-\beta H_\Lambda}\preceq \constlemm{\beta} e^{-\beta H_{\Lambda\setminus b}},
        \label{eq:lemmclaimedineq}
    \end{equation}
	    where $H_{\Lambda \setminus b}$ is obtained from $H_{\Lambda}$ by removing all interactions whose support includes site $b$ and $\constlemm{\beta}=\exp(\exp(40\,\widetilde\beta r^2))$.
    \label{lemm:cutgibbs}
\end{lemma}
	    \begin{proof}
	     By applying the same folding argument as in the proof of Theorem~\ref{th:tsEBD}, we relabel the qubits so that $b$ is at the left of the folded chain, and then we can use the bound on the norm of the Araki expansional, Eq.~(46), from Ref.~\cite{Bakshi2026}, with $\beta/2$ (and replacing $\mathfrak{K}\to 2r$, $\beta\to r\beta$), to give 
	    $\norm{e^{-\beta H_{\Lambda}/2}e^{\beta H_{\Lambda\setminus b}/2}}_{\infty}
	    \leq \constlemm{\beta}^{1/2}$ and, similarly, 
	    $\norm{e^{-\beta H_{\Lambda\setminus b }/2}e^{\beta H_{\Lambda}/2}}_{\infty}
	    \leq \constlemm{\beta}^{1/2}$. Therefore, all singular values of the
	    first operator lie between $\constlemm{\beta}^{-1/2}$ and
	    $\constlemm{\beta}^{1/2}$. Consequently,
	    \begin{equation*}
	        \constlemm{\beta}^{-1}\1
	        \preceq
	        e^{\beta H_{\Lambda\setminus b}/2}
	        e^{-\beta H_{\Lambda}}
	        e^{\beta H_{\Lambda\setminus b}/2}
	        \preceq
	        \constlemm{\beta}\1.
	    \end{equation*}
	    Conjugating this inequality by $e^{-\beta H_{\Lambda\setminus b}/2}$ yields Eq.~\eqref{eq:lemmclaimedineq}.
	    \end{proof}

\begin{proof}[Proof of Proposition~\ref{prop:omegaop}]
From Lemma~\ref{lm:staircase-separability}, we can expand $$(X_{n-1}\cdots X_0)(X_{n-1}\cdots X_0)^\dagger=\sum_\nu w_\nu\omega_1^{(\nu)}\otimes\cdots\otimes\omega_n^{(\nu)},$$ where $\omega_j^{(\nu)}$ are product density matrices with unit trace and $\omega_j^{(\nu)}\succeq 64^{-m}\1$. For the rest of the proof we fix a single label $\nu$ and omit the superindex $(\nu)$ (all operators defined next depend on $\nu$). Let
\begin{align}
G_j\coloneqq
    e^{-\frac{\beta}{2} H_{B_j}}
    \;\omega_j\;
    e^{-\frac{\beta}{2} H_{B_j}},&&&&
    G'_j\coloneqq 
64^{-m}\widetilde C_\beta^{-1}\;
e^{-\beta H_{[k_{j-1}+1,\bar k_j-1]}}
\otimes
e^{-\beta H_{[\bar k_j+1,k_j-1]}}.
\end{align}
Since $64^{-m} \1 \preceq \omega_j \preceq \1$, we have $64^{-m}e^{-\beta H_{B_j}} \preceq G_j\preceq e^{-\beta H_{B_j}}$, and by Lemma~\ref{lemm:cutgibbs}, we have ${G'}_j \preceq G_j\preceq e^{7m} {G'}_j,$ where we used $64^{m}\constlemm{\beta}^2\leq e^{7m}$. Let us define $\Delta_j=G_j-G_j'$, so that 
\begin{equation}
    0\preceq\Delta_j\preceq(1-e^{-7m})G_j.
\end{equation}

Expanding $\Omega=\sum_{\nu} w_\nu \bigotimes_{s} \bigotimes_{j\in\mathcal{S}_s} G_j$, we are to show that $\bigotimes_{j\in\mathcal{S}_s} G_j$ is a positive sum of positive operators with a cut across the midpoint of some block in $\mathcal{S}_s$.
Fix a superblock $\mathcal{S}$. Expanding every $\bigotimes_{j\in\mathcal{S}}G_j=\bigotimes_{j\in\mathcal{S}}(G_j'+\Delta_j)$ gives a positive sum, where every term containing a factor $G'_j$ has the desired product cut at the  midpoint of the $j$th block. However, there is a term $\bigotimes_{j\in\mathcal{S}}\Delta_j$ without such a cut. To deal with this term, pick a block $j_0\in \mathcal{S}$ and define $R=G_{j_0}'\otimes\bigotimes_{\smash{\substack{j\in\mathcal{S}\\ j\neq j_0}}}G_j$.  We expand 
\begin{align}
    \bigotimes_{j\in\mathcal{S}}G_j
    =&\Big[\bigotimes_{j\in\mathcal{S}}G_j -\bigotimes_{j\in\mathcal{S}}\Delta_j-R\Big]  + \Big[ R + \bigotimes_{j\in\mathcal{S}}\Delta_j\Big]
    \label{eq:streamlined-superblock-decomposition}
\end{align}
The first bracket is the positive sum of all products of $\Delta_j$ and $G_j'$ where the $j_0$th factor is $\Delta_{j_0}$ but not all the other factors are $\Delta_j$; in particular, some block is set to $G_{j}'$, which has a cut across its midpoint. The second bracket is separable across the midpoint of the block at $j_0$ by the following claim.

\noindent \textbf{Claim: }\textit{The positive operator $ R+\bigotimes_{j\in\mathcal{S}}\Delta_j$ is separable across the midpoint of the block at $j_0$.} 
    
\noindent Indeed, for each $j\neq j_0$, diagonalize 
\begin{equation}
   G_j^{-1/2}\Delta_j
   G_j^{-1/2}
    =
    \sum_{t_j}\lambda_{j,t_j}P_{j,t_j},
    \qquad
    0\leq\lambda_{j,t_j}\leq1-e^{-7m}.
\end{equation}
Then, we can write 
\begin{align*}
    R
    +\bigotimes_{j\in\S}\Delta_j
    &=
   G_{j_0}'^{1/2}
    \bigg(
        \1_{j_0}\otimes
        \bigotimes_{j\neq j_0}G_j
        +
        G_{j_0}'^{-1/2}
        \Delta_{j_0}
        G_{j_0}'^{-1/2}
        \otimes
        \bigotimes_{j\neq j_0}\Delta_j
    \bigg)
    G_{j_0}'^{1/2} \\
    &= G_{j_0}'^{1/2}
    \sum_{\bm t}
    \left(
        \1_{j_0}+
        K_{\bm t}
    \right)G_{j_0}'^{1/2}\otimes 
    E_{\bm t}
\end{align*}
with 
\begin{equation}
    E_{\bm t}=\bigotimes_{j\neq j_0}
    \left(
        G_j^{1/2}
        P_{j,t_j}
        G_j^{1/2}
    \right),\quad K_{\bm t}=\Big(\prod_{j\neq j_0}\lambda_{j,t_j}\Big)
        G_{j_0}'^{-1/2}
        \Delta_{j_0}
        G_{j_0}'^{-1/2},
        \label{eq:Kadecomp}
\end{equation} and $\bm t=(t_j)_{j\neq j_0}$. Since $G_{j_0}'^{1/2}$ is a product operator across the midpoint of the block, we only need to show that each operator $\1+K_{\bm t}$ is separable across that cut.

We show that in fact $\1+K_{\bm t}$ is completely separable within the $j_0$ block. This follows from a general result: if
\(K=K^\dagger\) is supported on at most \(\ell\) qubits and
\(\|K\|_\infty\leq4^{-\ell}\), then \(\1+K\) is a positive combination of stabilizer product states. The proof is simple: expanding
\(K=\sum_P\alpha_P P\) in the Pauli basis, one has
\(\abs{\alpha_P}\leq\|K\|_\infty\). Since there are at most
\(4^\ell\) Pauli strings, the condition
\(\|K\|_\infty\leq4^{-\ell}\) allows one to express \(\1+K\) as a
nonnegative combination of operators \(\1+c_PP\), with
\(\abs{c_P}\leq1\), each of which is diagonal with nonnegative
coefficients in a stabilizer product basis. See Lemma~6.2 of Ref.~\cite{Bakshi2026} for a full statement and proof. Since the Hermitian operators $K_{\bm t}$ are supported on the \(m\) qubits of the
block at \(j_0\), we just need to show that their operator norm is bounded by $4^{-m}$. Indeed,
\begin{align}
    \norm{K_{\bm t}}_\infty&=\Big(\prod_{j\neq j_0}\lambda_{j,t_j}\Big)\big \|
        G_{j_0}'^{-1/2}
        \Delta_{j_0}
        G_{j_0}'^{-1/2}\big \|_\infty
    \leq
    \left(1-e^{-7m}\right)^{L-1}
    \left(e^{7m}-1\right)
    \notag
    \\&=
    e^{7m}
    \left(1-e^{-7m}\right)^L
    \notag
    \leq
    \exp(
        7m-e^{2m}
    )
    \notag
    \leq
    4^{-m}.
\end{align}
Here the first inequality uses $
    G_{j_0}'^{-1/2}
    \Delta_{j_0}
    G_{j_0}'^{-1/2}
    \preceq
    (e^{7m}-1)\1$, 
the second uses \(1-x\leq e^{-x}\) and
\(L=e^{9m}\), and the last one holds since \(m> 2\).
\end{proof} 

\section{Polynomial-time sampling algorithm}
\label{app:alg}
Here, we prove Theorem~\ref{th:02} by describing the classical sampling algorithm in detail. The algorithm follows the proof of Theorem~\ref{th:01} step by step: it first samples a term $\nu$ of the separable decomposition of the staircase of QPIs, then samples the position of the cut in each superblock, and finally samples pure states on the resulting blocks. We begin by fixing notation.

Recall that the system comprising $N$ qubits has been split up into $n=N/m$ elementary blocks of size $m$ each, and that we group those blocks into $q=\floor{n/L}$ superblocks comprising the blocks with labels $\S_s=[(s-1)L+1, sL]$. If $n>qL$, we include the remaining blocks in the last superblock.
Following Corollary~\ref{lm:BCP-mod}, we write
\begin{equation}
    g_\beta \propto M\Omega M^\dagger,
    \qquad
    M\coloneqq\prod_{j=0}^{n-1}M_j,
    \qquad
    \Omega
    =
    \sum_\nu w_\nu
    \bigotimes_{j=1}^nG_j^{(\nu)},
    \qquad
    G_j^{(\nu)}
    \coloneqq
    e^{-\frac{\beta}{2} H_{B_j}}\omega_j^{(\nu)}e^{-\frac{\beta}{2} H_{B_j}}.
    \label{eq:definitions-S4}
\end{equation}
Here, \(w_\nu\) are the nonnegative coefficients produced by the
separable decomposition of \(Y\) in Lemma~\ref{lm:staircase-separability}. Defining
\begin{equation}
    \widetilde\rho_\nu\coloneqq M\left(\bigotimes_jG_j^{(\nu)}\right)M^\dagger,
    \qquad
    \rho_\nu\coloneqq\frac{\widetilde\rho_\nu}{\tr(\widetilde\rho_\nu)},
    \qquad
    \pi_\nu\coloneqq\frac{w_\nu\tr(\widetilde\rho_\nu)}
    {\sum_{\nu'} w_{\nu'}\tr(\widetilde\rho_{\nu'})},
    \label{eq:rho-nu}
\end{equation}
we have \(g_\beta=\sum_\nu\pi_\nu\rho_\nu\).

The distribution \(\pi_\nu\) can be sampled by adapting the algorithm
of Sec.~9 of Ref.~\cite{Bakshi2026}, but instead following the construction of Lemma~\ref{lm:staircase-separability}. One applies the sampling and backtracking procedure to the decomposition tree, including the  operators \(M_j\) and $\exp(-\beta H_{B_j}/2)$ in the trace weights. Since every step
acts on a number of qubits depending only on \(\beta\) and \(r\), the
resulting distribution \(\pi_\nu\) can be sampled to error
\(\varepsilon\) in \(\poly(N,1/\varepsilon)\) time. 

Once $\nu$ is fixed, the second step is to apply the positive superblock decomposition, which allows us to write $\rho_\nu$ as a mixture of states that each have at least one cut in each superblock. For a superblock $\mathcal{S}_s=\{j_0<j_1<\dots<j_{L-1}\}$ and $b\in [0,L-1]$, define
\begin{equation}
Q_{s,b}
=
\begin{cases}
G_{j_0}'
\bigotimes_{c=1}^{L-1}G_{j_c} + \bigotimes_{c=0}^{L-1}\Delta_{j_c},& b=0,\\
\left(\bigotimes_{c=0}^{b-1}\Delta_{j_c}\right)\otimes G_{j_b}'\otimes\left(\bigotimes_{c=b+1}^{L-1}G_{j_c}\right),& b=1,\ldots,L-1.
\end{cases}
\label{eq:Q-b}
\end{equation}
such that $\bigotimes_{j\in\mathcal S_s}G_j
    =
    \sum_{b=0}^{L-1}Q_{s,b}.$ The label $b$ records the first block $j_b$ at which a factor $G'$ appears (the term with no cuts has been added to the $b=0$ term).
For \(b\geq1\), the operator $Q_{s,b}\eqqcolon L_{s,b}\otimes R_{s,b}$
 factorizes across the midpoint of the \(j_{b}\)th elementary
block. For \(b=0\), we have shown in App.~\ref{app:proof-of-th-01} that the operator \(Q_{s,0}\) is separable across the midpoint of the $j_0$th block,
\begin{equation}
    Q_{s,0}
    =
    \sum_{\alpha_s}
    L_{s,0,\alpha_s}
    \otimes
    R_{s,0,\alpha_s}.
\end{equation}
The labels \(\alpha_s\) and the corresponding positive product
operators can be constructed explicitly by expanding the operators
\(\1+K_{\bm t}\) of Eq.~\eqref{eq:Kadecomp} in the Pauli basis and decomposing the resulting terms into product operators. The label \(\alpha_s\) includes the tuple \(\bm t\) and all further choices in this separable decomposition.
For each superblock \(s\), let \(a_s=(b_s,\alpha_s)\), where \(b_s\) specifies the selected cut and \(\alpha_s\) labels a term in the separable decomposition of \(Q_{s,0}\) when \(b_s=0\). For \(b_s\neq0\), \(\alpha_s\) takes a single value, and we set \(L_{s,a_s}=L_{s,b_s}\) and \(R_{s,a_s}=R_{s,b_s}\).
Writing \(\vec a=(a_1,\ldots,a_q)\), the pair \((\nu,\vec a)\) specifies the label \(\mu\) in Proposition~\ref{prop:omegaop}. Conditional on \(\nu\), we want to sample \(\vec a\) from the joint distribution
\begin{align}
p_\nu(\vec a)
\propto 
\tr\!\left[M\bigotimes_{s=1}^{q}\left(
L_{s,a_s}\otimes R_{s,a_s}\right)M^\dagger\right]
=\displaystyle\prod_{s=1}^{q}F_s(a_s,a_{s+1}),
\label{eq:p-b}
\end{align}
where
\begin{equation*}
    F_s(a_s,a_{s+1})
\coloneqq
\tr\!\left[
M_{a_s,a_{s+1}}
\left(
R_{s,a_s}
\otimes
L_{s+1,a_{s+1}}
\right)
M_{a_s,a_{s+1}}^\dagger
\right],\qquad M_{a_s,a_{s+1}}\coloneqq
\prod_{\substack{\,
j_{b_{s}}(s)\leq \ell<j_{b_{s+1}}(s+1)}}
M_\ell,
\end{equation*}
the superblock indices are interpreted cyclically 
\(q+1\to 1\),
and we omit the $\nu$ dependence for brevity.

Equation~\eqref{eq:p-b} shows that the labels form a classical nearest-neighbor distribution on a ring: the weight of a global configuration is a product of factors, each depending only on two consecutive superblocks.
Such a probability distribution can be sampled from in linear time using a standard forward-backward algorithm~\cite[Sec.~17.4.5]{Murphy2012}.
For each initial label \(a_1\), we work forwards to compute the total
weight of all configurations between \(a_1\) and \(a_s\), then we sample one \(a_1\) from the resulting total distribution,
and walk backwards to sample each preceding label.
Specifically, the effective weight contributed by configurations with boundary labels
\(a_1\) and \(a_s\) is given by
\begin{equation}
    \Lambda_s^{(a_1)}(a_s)
    =
    \sum_{a_{s-1}}
    \Lambda_{s-1}^{(a_1)}(a_{s-1})
    F_{s-1}(a_{s-1},a_s),
    \qquad
    \Lambda_2^{(a_1)}(a_2)
    =
    F_1(a_1,a_2),
    \label{eq:Lambda}
\end{equation}
for \(s=3,\ldots,q\), and the total weight of all configurations with a
fixed value of \(a_1\) is $\Lambda_1(a_1) \coloneqq\sum_{a_q}\Lambda_q^{(a_1)}(a_q) F_q(a_q,a_1).$

Now we sample
\(a_1\) according to $p_\nu(a_1)
    =
    {\Lambda_1(a_1)}/
    {\sum_{a_1'}\Lambda_1(a_1')} $ and walk backwards:  for each
\(s=q,\ldots,2\), we sample $a_{s}$ 
\begin{equation}
    p_\nu(a_{s}\mid a_1,a_{s+1},\ldots,a_q)
    =
    \frac{
        \Lambda_{s}^{(a_1)}(a_{s})
        F_{s}(a_{s},a_{s+1})
    }{
        \Lambda_{s+1}^{(a_1)}(a_{s+1})
    },
\end{equation}
 with $\Lambda_{q+1}=\Lambda_{1}$. Since the number of possible
values of each \(a_s\) depends only on \(\beta\) and \(r\), both the
forward and backward passes take \(O(N)\) time, and the sampling is exact.
Finally, by spectrally decomposing the elementary-block factors of the sampled operators $R_{s,a_s}\otimes L_{s+1,a_{s+1}}$, applying the corresponding operators $M_j$, and sampling the resulting pure states with probabilities proportional to their squared norms, we obtain the product of matrix product states over blocks appearing in Theorem~\ref{th:02}. Including the final spectral choices gives the refined label $\mu$ used in Theorem~\ref{th:01}.

\section{Proof of Lemma~\ref{lm:staircase-separability}}

\label{app:periodic-staircase}
Here we prove Lemma~\ref{lm:staircase-separability}, adapting the pinning technique of Refs.~\cite{Bakshi2024,Bakshi2026} to two-sided QPIs and periodic boundary conditions. We first establish two auxiliary lemmas.
\begin{lemma}[Variation of pinning Lemma 6.4 in \cite{Bakshi2026}]
	Let $-1\leq c\leq 1$, $n,k\in\NN$ with $n>k$ and $P=P_1\otimes \cdots \otimes P_n$ be a Pauli string whose Pauli weight on $[k]$ we call $q_k$. Then, for any positive $\lambda<1$,
	\begin{equation}
		\1+cP = \sum_{\sigma_1,\dots,\sigma_k=\pm 1}\omega_{[1,k]}^{({\sigma_1},\dots ,{\sigma_k})}
		\otimes \left[\1+c'_{{\sigma_1,\dots,\sigma_k}}P_{[k+1,n]}\right],
		\label{eq:pinning}
	\end{equation}
    where $c'_{{\sigma_1,\dots,\sigma_k}}=c\left(\prod_{j\leq k:P_j\neq\1}\sigma_j\right)\lambda^{-q_k}$, $\omega_{[1,k]}^{({\sigma_1},\dots ,{\sigma_k})}\succeq \eta\1$ is a product state with  $\eta\coloneqq 2^{-k}(1-\lambda)^{q_k}$.
	\label{lm:pinning}
\end{lemma}
\begin{proof}
	To arrive at the result, we iterate the following decomposition:
	If $P_1=\1$, we have $\1+cP=\sum_{\sigma\in\{\pm1\}}(\1/2)\otimes(\1+cP_{[2,n]})$.
	Otherwise, we pick a $\lambda$ and use the identity
	\begin{equation}
		\begin{aligned}
			\1 + cP &= \frac{1}{2}\left[ 
				(\1+\lambda P_1)\otimes\left(\1+(c/\lambda) P_{[2,n]}\right)+
				(\1-\lambda P_1)\otimes\left(\1-(c/\lambda) P_{[2,n]}\right)
			\right]\\
					&= \sum_{\sigma=\pm1} \omega_\sigma^{P_1}\otimes(\1+ \sigma (c/\lambda) P_{[2,n]}),
		\end{aligned}
	\end{equation}
	where $\omega_\sigma^{P_1} \succeq \frac{1}{2}(1-\lambda)\1$ and $\tr(\omega_\sigma^{P_1})=1$.

	Iterating, we arrive at the desired decomposition,
	\begin{equation}
		\1 + cP = \sum_{\sigma_1,\dots,\sigma_k=\pm 1}
		\omega_{\sigma_1}\otimes\cdots\otimes\omega_{\sigma_k}\otimes
		\left[\1+\left(\prod_{j\leq k:P_j\neq\1}\frac{\sigma_j}{\lambda}\right)cP_{[k+1,n]}\right],
	\label{eq:pinning-decomposition}
	\end{equation}
	where  $\omega_{\sigma_j}=\1/2$ for identity sites.
	The state on the first $k$ sites is clearly a product state with unit trace and by construction each $\omega_{\sigma_j}\succeq (1-\lambda)\1/2$ (or $\succeq \1 /2$ for identity sites).
	Consequently, if $q_k$ is the weight of the Pauli string $P$ on the first $k$ sites, we have 
	\begin{equation*}
\omega_{\sigma_1}\otimes\cdots\otimes\omega_{\sigma_k}\succeq 2^{-k}(1-\lambda)^{q_k}\1.
	\qedhere
	\end{equation*}
\end{proof}

\begin{lemma}[Variant of Lemma 6.5 in \cite{Bakshi2026}]
	Let $A$ be a Pauli string with $\supp A\subseteq[a]$, $c\in\RR,\ |c|\leq \delta^a$, let $X^{(k)}$ be a QPI with decay $\gamma$ anchored at $k$, and suppose that $8\gamma\leq\delta\leq\lambda/7$ for $ 0<\lambda<1.$ 
	Let $\P_{k,\ell}$ be the set of Pauli strings supported on $[k-\floor{\ell/2},k+\ceil{\ell/2}-1]$.
	Then
	\begin{equation}
		 X^{(k)}(\1+cA)X^{(k)\dagger}
		=\sum_{\nu=1}^J w_\nu\omega^{(\nu)}\otimes(\1+c_\nu B_\nu),
		\label{eq:sandwich}
	\end{equation}
	where $J$ is a finite integer,
	$\omega^{(\nu)}$ are product states with unit trace and $\omega^{(\nu)}\succeq ((1-\lambda)/2)^k\1 $,
	$B_\nu$ are Pauli strings supported in $[k+1,b_\nu+k]$, $w_\nu$ are nonnegative weights,
	and $c_\nu\in\RR, |c_\nu|\leq \delta^{b_\nu}$.
	\label{lm:sandwiched-pinning}
\end{lemma}
\begin{proof}
	 Let $p_\ell\coloneqq
    \frac{2^{-\ell}}{1-2^{-N}},$ and define $Z_\ell\coloneqq
    \1+ F_\ell^{(k)}/p_\ell.$ By decomposing the QPI into Pauli strings we get
	\begin{equation}
		X^{(k)}=\1 + \sum_{\ell=1}^N F^{(k)}_\ell
		= \sum_{\ell=1}^N p_\ell Z_\ell
		=\sum_{\ell=1}^N p_\ell\sum_{P\in\P_{{k,\ell}}}4^{-\ell}(\1+p_\ell^{-1}4^{\ell}\alpha_P^{(\ell)}P).
		\label{eq:X-decomposition}
	\end{equation}
	Since $\norm{F_\ell}\leq\gamma^\ell$, we have $|\alpha_P^{(\ell)}|\leq \gamma^\ell$.
	Using this decomposition, we write
	\begin{equation*}
		\begin{aligned}
 X^{(k)}(\1+cA)X^{(k)\dagger}
    &=
    \frac12
    \sum_{\ell,\ell'=1}^N
    p_\ell p_{\ell'}4^{-\ell-\ell'}
    \sum_{P\in\P_{k,\ell}}
    \sum_{P'\in\P_{k,\ell'}}
    \Big[
        \left(
            \1+p_{\ell}^{-1}4^\ell
            \alpha_P^{(\ell)}P
        \right)
        (\1+cA)
    \\
    &\hspace{120pt}\times
        \left(
            \1+p_{\ell'}^{-1}4^{\ell'}
            \alpha_{P'}^{(\ell')}P'
        \right)^\dagger
        +\mathrm{H.c.}
    \Big]
    \\
    &=
    \sum_{\ell,\ell'=1}^N
    p_\ell p_{\ell'}4^{-\ell-\ell'}
    \sum_{P\in\P_{k,\ell}}
    \sum_{P'\in\P_{k,\ell'}}
    \sum_{j=1}^7
    \left(\frac{\1}{7}+O_j^{\ell\ell'}\right),
    \end{aligned}
	\end{equation*}
	where in the first line we used the Hermiticity of the left-hand side to explicitly make each summand Hermitian,
	and to go to the second line, we have expanded the summand and collected the resulting terms $O_j$.
	We then use \cref{lm:pinning} on each of these terms, to pin sites 1 through $k$.
	This ensures the lower bound $\omega^{(\nu)}\succeq ((1-\lambda)/2)^k\1$,
	but it remains to check that $|c_\nu|\leq \delta^{b_\nu}$, which we will do term by term. Here the index is $\nu=(\ell,\ell',P,P',j,\sigma_1,\dots,\sigma_k)$. Henceforth we leave the dependence on all variables except $j$ implicit, and use $j$ instead of $\nu$.

	Introducing the shorthand $z=p_\ell^{-1}4^\ell\alpha_P^{(\ell)},$ 
    $z'=p_{\ell'}^{-1}4^{\ell'}
    \alpha_{P'}^{(\ell')},$ 
	the $O_j$ terms read
	\begin{equation}
		\begin{aligned}
			O_1 &= cA,\quad
			O_2 = \Re(z)P,\quad
			O_3=\Re(z')P',\quad
			O_4=czPA/2+\text{H.c.},\\
			O_5 &= cz'P'A/2+\text{H.c.}, \quad
			O_6 =zz^{\prime*}PP'/2+\text{H.c.},\quad
			O_7=czz^{\prime*}PAP'/2+\text{H.c}.
		\end{aligned}
		\label{eq:Os}
	\end{equation}
	Note that these all evaluate to the form $(\text{coefficient})\times(\text{Pauli string})$,
    i.e., $O_j=x_j R_j$.
    The coefficients obey the bound $|x_j|\leq \delta^{D_j}$ with 
    \begin{equation}
        D_j = (a, \ell, \ell', a+\ell, a+\ell', \ell+\ell', \ell+\ell'+a).
    \end{equation}
	Let $b_j$ be the smallest nonnegative integer such that $\supp((R_j)_{[k+1,N]})\subseteq[k+1,k+b_j]$.
    By assumption, the Pauli strings $P$ and $P'$ have support in $[1,s_{k,\ell}]$ and $[1,s_{k,\ell'}]$, where $s_{k,\ell}=k+\ceil{\ell/2}-1$. Set $r_\ell\coloneqq s_{k,\ell}-k=\ceil{\ell/2}-1$.
    We therefore have
    \begin{equation}
        b_j\leq(\max(0,a-k),r_\ell,r_{\ell'},\max(0,a-k,r_\ell),\max(0,a-k,r_{\ell'}),\max(r_\ell,r_{\ell'}),\max(0,a-k,r_\ell,r_{\ell'})).
    \end{equation}
   Let $q_j$ be the Pauli weight of $R_j$ within $[1,k]$.
   Note that $q_j+b_j\leq D_j$ and if $q_j=0$, $b_j\leq D_j-1$.
   In applying \cref{lm:pinning}, the coefficient $c_j=7x_j$ gets blown up at most by a factor $\lambda^{-q_j}$, which gives us the bound on the coefficient after pinning
   \begin{align*}
       |c_j'|
       &\leq 7\lambda^{-q_j}|x_j|
       \leq 7\lambda^{-q_j}\delta^{D_j-b_j}\delta^{b_j}
       \leq \delta^{b_j}.
   \end{align*}
   The last inequality follows from $D_j-b_j\geq q_j$ and $7(\delta/\lambda)^{q_j}\leq1$ if $q_j\geq1$, and from $D_j-b_j\geq1$ and $7\delta\leq\lambda<1$ if $q_j=0$.
\end{proof}
We now prove Lemma~\ref{lm:staircase-separability} by iterating Lemma~\ref{lm:sandwiched-pinning}. 
\begin{proof}[Proof of Lemma~\ref{lm:staircase-separability}]
We will iterate Lemma~\ref{lm:sandwiched-pinning}. However, for the first QPI $X_0$ we proceed separately, since it may have periodic decay which crosses from qubit $1$ into qubit $N$. Set $\mathcal{X}\coloneqq X_{n-1}\cdots X_1$. As in the proof of Theorem~\ref{th:tsEBD}, relabel the qubits $(1,N,2,N-1,\ldots)$, so that $X_0$ is a one-sided QPI anchored at the first site. Applying Lemma~\ref{lm:sandwiched-pinning} with $k=1$, $c=0$, $\delta=1/14$, and $\lambda=1/2$, followed by Lemma~\ref{lm:pinning} gives
\begin{equation}
	X_0X_0^\dagger
	=
	\sum_u w_u\tau_u,
	\qquad
	\tau_u=\bigotimes_{x=1}^N\tau_{u,x},
\end{equation}
where the $\tau_{u,x}$ are single-qubit density matrices satisfying
$\frac14\1\preceq\tau_{u,x}\preceq\frac34\1.$
For each $u$, define
\begin{equation}
	\widetilde X_{j,u}
	\coloneqq
	\tau_u^{-1/2}X_j\tau_u^{1/2},
	\qquad
	\widetilde{\mathcal{X}}_u
	\coloneqq
	\widetilde X_{n-1,u}\cdots\widetilde X_{1,u}.
\end{equation}
Then $\mathcal{X}\tau_u \mathcal{X}\dagg
	=
	\tau_u^{1/2}\widetilde{ \mathcal{X}}_u\widetilde{ \mathcal{X}}_u\dagg\tau_u^{1/2}.$ 
The similarity transformation preserves the supports of the QPI terms, and
\begin{equation}
	\left\|
	\tau_u^{-1/2}F_{j,\ell}\tau_u^{1/2}
	\right\|_\infty
	\leq
	3^{\ell/2}\left(\frac1{112}\right)^\ell
	=
	\left(\frac{\sqrt3}{112}\right)^\ell.
\end{equation}
 Iterating Lemma~\ref{lm:sandwiched-pinning} with parameters $\gamma=\frac{\sqrt3}{112}$, $\delta=\frac18$, and $\lambda=\frac78$ along the open chain $[2,N]$ yields
\begin{equation}
	\widetilde{\mathcal{X}}_u\widetilde {\mathcal{X}}_u\dagg
	=
	\sum_v w_{u,v}
	\bigotimes_{j=1}^n\widetilde\omega_j^{(u,v)},
	\qquad
	\widetilde\omega_j^{(u,v)}
	\succeq
	16^{k_{j-1}-k_j}\1,
\end{equation}
where the identity on site $1$ is included in the last cyclic block. Conjugating by the product operator $\tau_u^{1/2}$ preserves this tensor-product structure and gives block factors $\widehat\omega_j^{(u,v)}
	\succeq
	64^{k_{j-1}-k_j}\1.$  Normalizing the block factors and absorbing their traces into the weights gives the desired decomposition. Collecting $(u,v)$ into a single label $\nu$ yields the claimed decomposition.
\end{proof}
\section{Proof of Corollary~\ref{th:thermalization}}
\label{app:therm-proof}

Here, we prove Corollary~\ref{th:thermalization}. We use two results from Ref.~\cite{PilatowskyCameo2026}, which we state here in 1D, together with a bound on the subsystem purity of 1D Gibbs states (Lemma~\ref{lm:subsystem-purity}).
\begin{lemma}[EDGE theorem of Ref.~\cite{PilatowskyCameo2026}, 1D]
\label{th:EDGE}
Let \(H\) be a translation-invariant, geometrically local Hamiltonian on a 1D ring, and suppose that \(H\) has
nondegenerate spectral gaps. Assume that an ensemble
\(\mathcal{E}_\beta\) satisfies
\begin{equation}
    \norm{
        \mathbb{E}_{\ket{\psi}\sim\mathcal{E}_\beta}
        [\dyad{\psi}]
        -
        g_\beta
    }_1
    \leq
    2^{-N}\varepsilon_{\mathrm{GE}}(N) \qquad\text{and}\qquad
\mathbb{E}_{\ket{\psi}\sim\mathcal{E}_\beta}
    \Big[
        \sum_j
        \abs{\braket{E_j|\psi}}^4
    \Big]
    \leq
    \varepsilon_{\mathrm{ED}}(N),
\label{eq:conditionsofEdgetheom}
\end{equation}
where \(\{\ket{E_j}\}_j\) is the energy eigenbasis. For every constant-size region
\(A\),
\begin{equation}
    \mathbb{E}_{\ket{\psi}\sim\mathcal{E}_\beta}
    \mathbb{E}_t
    \left[
        \norm{\psi_A(t)-g_{\beta,A}}_1
    \right]
    \leq
    O(N^{-\gamma})
    +
    2\varepsilon_{\mathrm{GE}}(N)
    +
    2^{|A|}
    \sqrt{\varepsilon_{\mathrm{ED}}(N)}
    \label{eq:EDGE-theorem-bound}
\end{equation}
for every \(\gamma<1/2\).
\end{lemma}

\begin{lemma}[Theorem~3 of Ref.~\cite{PilatowskyCameo2026}]
\label{th:th03ofPil2026}
Let \(\mathcal{E}\) be an ensemble of pure states with circuit-depth complexity which is at most subpolynomial in $N$, and let $\rho
    \coloneqq
    \mathbb{E}_{\ket{\psi}\sim\mathcal{E}}
    [\dyad{\psi}]$
be its average state. Suppose that \(\rho\) is translation invariant,
has finite correlation length, and satisfies
\begin{equation}
    \tr(\rho_R^2)
    \leq
    e^{-\Omega(|R|)}
    \label{eq:small-subsystem-purity}
\end{equation}
for every interval \(R\). Then, in any basis
\(\{\ket{j}\}_j\) consisting of translation eigenstates, and for any  \(\eta>0\),
\begin{equation}
    \mathbb{E}_{\ket{\psi}\sim\mathcal{E}}
    \Big[
        \sum_j
        \abs{\braket{j|\psi}}^4
    \Big]
    \leq
    O(N^{-1+\eta}).
    \label{eq:average-IPR-theorem}
\end{equation}
\end{lemma}
Now consider any ensemble $\mathcal{E}_{\beta,\mathcal{C}}$ with the specified properties  (P1--3). By property (P3), the average state $\rho_{\mathcal{E}}=\mathbb{E}_{\ket{\psi}\sim \mathcal{E}_{\beta,\mathcal{C}}}[\dyad{\psi}]$ has maximum entropy among the average states of all other ensembles satisfying (P1) and (P2). In particular the ensemble supplied by Theorem~\ref{th:01} satisfies these properties and has the Gibbs state as an average state. Furthermore, the Gibbs state is the unique maximally entropic state subject to the energy constraint $\tr(\rho H)=\tr(g_\beta H)$. Thus, by (P3), the average state of any ensemble satisfying (P1--3) is also the Gibbs state $\rho_{\mathcal{E}}=g_\beta$, i.e., $\varepsilon_{\mathrm{GE}}(N)=0$ in Eq.~\eqref{eq:conditionsofEdgetheom}. 

Lemma~\ref{th:th03ofPil2026} combined with Lemma~\ref{th:EDGE} gives Corollary~\ref{th:thermalization}. We just need to verify that the conditions of Lemma~\ref{th:th03ofPil2026} are satisfied for the Gibbs state. The Gibbs
state \(g_\beta\) is translation invariant because \(H\) is
translation invariant, and it has a finite correlation length at every finite inverse temperature~\cite{Araki1969}. Finally, Eq.~\eqref{eq:small-subsystem-purity} follows from the following:
\begin{lemma}[Exponentially small subsystem purity]
\label{lm:subsystem-purity}
Let \(H\) be a local translation-invariant Hamiltonian on a
one-dimensional ring, and let \(g_\beta\) be its Gibbs state at a fixed
finite inverse temperature \(0\leq\beta<\infty\). Then, for every
interval \(R\),
\begin{equation}
    \tr(g_{\beta,R}^2)
    \leq
    e^{-\Omega(|R|)}.
    \label{eq:exponentially-small-subsystem-purity}
\end{equation}
\end{lemma}

\begin{proof}
Let
\begin{equation}
    g_\beta^R
    \coloneqq
    \frac{e^{-\beta H_R}}{Z_{\beta,R}},
    \qquad
    Z_{\beta,R}
    \coloneqq
    \tr(e^{-\beta H_R}),
\end{equation}
be the Gibbs state of the Hamiltonian restricted to $R$. If $R$ is not the full chain, let 
$x,y\notin R$ be the two sites adjacent to $R$. Applying
Lemma~\ref{lemm:cutgibbs} first at $x$ and then at $y$ gives
\begin{equation}
    \widetilde C_\beta^{-2}
    e^{-\beta H_{\setminus\{x,y\}}}
    \preceq
    e^{-\beta H}
    \preceq
    \widetilde C_\beta^{2}
    e^{-\beta H_{\setminus\{x,y\}}}.
    \label{eq:two-cut-gibbs-comparison}
\end{equation}
Taking traces and normalizing gives
\begin{equation}
    \widetilde C_\beta^{-4}
    \frac{e^{-\beta H_{\setminus\{x,y\}}}}
         {\tr(e^{-\beta H_{\setminus\{x,y\}}})}
    \preceq
    g_\beta
    \preceq
    \widetilde C_\beta^{4}
    \frac{e^{-\beta H_{\setminus\{x,y\}}}}
         {\tr(e^{-\beta H_{\setminus\{x,y\}}})}.
\end{equation}
The Hamiltonian $H_{\setminus\{x,y\}}$ has no interactions between
$R$ and its complement. Taking the partial trace over the complement of $R$ therefore
gives $\widetilde C_\beta^{-4}g_\beta^R
    \preceq
    g_{\beta,R}
    \preceq
    \widetilde C_\beta^{4}g_\beta^R.$
Consequently,
\begin{equation}
    \tr(g_{\beta,R}^2)
    \leq
    \widetilde C_\beta^{8}
    \tr[(g_\beta^R)^2].
    \label{eq:subsystem-purity-from-global-purity}
\end{equation}
Since $\tr[(g_\beta^R)^2]=Z_{2\beta,R}/Z_{\beta,R}^2$,
the existence of the pressure $p(\beta)\coloneqq
    \lim_{|R|\to\infty}{\log Z_{\beta,R}}/{|R|}$ by Theorem~2.1(i) of Ref.~\cite{Araki1969}
gives
\begin{equation}
    s_2(\beta)\coloneqq -\lim_{|R|\to\infty}
    \frac{1}{|R|}\log\tr[(g_\beta^R)^2]
    = 2p(\beta)- p(2\beta).
\end{equation}
We just need to show that $s_2(\beta)>0$. Indeed, again Theorem~2.1(i) of Ref.~\cite{Araki1969} implies that $p(\beta)$, and hence $s_2(\beta)$,
is real analytic at every finite inverse temperature.
Moreover, $s_2(\beta)\geq0$, and convexity of the pressure gives
$s_2'(\beta)=2(p'(\beta)-p'(2\beta))\leq0$. Then, $s_2(\beta)$ is a nonnegative, nonincreasing analytic function on $[0,\infty)$, and such a function either vanishes everywhere or does not vanish at all. Since $s_2(0)=\log2$, we get $s_2(\beta)>0$, as required.
\end{proof}
\section{Fermionic systems}
\label{app:ferm}

Here, we prove Theorem~\ref{th:03}. We use a Jordan--Wigner transformation, which preserves locality in 1D~\cite{Mbeng2024}. Specifically, we identify the fermionic Fock space with the Hilbert space of $N$ spins via  $$c_j=\Big(\prod_{k<j}Z_k\Big)\Big(\frac{X_j+iY_j}{2}\Big),$$ where $X_j,Y_j,Z_j$ denote the Pauli matrices on qubit $j$. Upon applying this transformation to the fermionic Hamiltonian, we obtain two local spin Hamiltonians $H_{\pm}$, one for each fermion-parity sector. Explicitly, the total parity operator $P=(-1)^{\sum_j c_j^\dagger c_j}=\prod_{j=1}^N Z_j$   diagonalizes into even and odd sectors $\Pi_\pm=\frac{\1\pm P}{2}$, and the total Hamiltonian $H$ expands as \cite{Mbeng2024}
\begin{equation}
    H=\Pi_+H_+\Pi_+ + \Pi_- H_-\Pi_-.
\end{equation}
The Hamiltonians $H_\pm$ commute with $P$ and have
the same interaction range and local term norm bounds
as the original Hamiltonian.

Applying Theorem~\ref{th:01} to each $H_\sigma$ separately gives
\begin{equation}
g_\beta=\frac{1}{\tr(e^{-\beta H})}\sum_{\sigma=\pm}
\Pi_\sigma e^{-\beta H_\sigma}\Pi_\sigma=\sum_{\sigma=\pm,\mu}\frac{\tr(e^{-\beta H_\sigma})}{\tr(e^{-\beta H})}p_{\sigma\mu}\Pi_\sigma\dyad{\psi_{\sigma\mu}}\Pi_\sigma=\sum_{\sigma=\pm,\mu}p'_{\sigma\mu}\dyad{\phi_{\sigma\mu}},
\end{equation}
where  
\begin{align}
\ket{\phi_{\sigma\mu}}=\frac{\Pi_\sigma\ket{\psi_{\sigma\mu}}}{\|\Pi_\sigma\ket{\psi_{\sigma\mu}}\|},&& p'_{\sigma\mu}=\frac{\tr(e^{-\beta H_\sigma})}{\tr(e^{-\beta H})}p_{\sigma\mu}\|\Pi_\sigma\ket{\psi_{\sigma\mu}}\|^2.
\end{align}
The pair $(\sigma,\mu)$ plays the role of the label $\mu$ in Theorem~\ref{th:03}.

We now show that the states $\ket{\phi_{\sigma \mu}}$ are obtained from a constant-depth circuit on a Gaussian state. Indeed, each $\ket{\psi_{\sigma\mu}}
=\bigotimes_{s=1}^{q}
\ket{\psi_{\sigma\mu s}}_{A_s^{(\sigma,\mu)}}$
is a product over contiguous blocks of at most $C_{\beta}$ sites. Henceforth we fix $\sigma$ and $\mu$ and drop the indices and superindices. Within each block we split into the two parity sectors. Let $P_{s}=(-1)^{\smash{\sum_{j\in A_s} c_j^\dagger c_j}}=\prod_{j\in A_s}Z_j$ be the parity operator for modes in $A_s$, and let $\Pi_{s,\pm}$ be its projectors onto the even and odd eigenspaces.
Split $\ket{\psi_{s}}=a_{s+} \ket{\psi_{s}^{\smash{(+)}}}+a_{s-} \ket{\psi_{s}^{\smash{(-)}}}$, where $a_{s\pm} \ket{\psi_{s}^{\smash{(\pm)}}}=\Pi_{s\pm}\ket{\psi_{s}}$. There exists a parity-preserving unitary $U_{s}$ (meaning that $[U_{s},P_s]=0$), such that
\begin{align}
    U_{s}\ket{00\cdots 0}=\ket{\psi_{s}^{(+)}}, && &&U_{s}\ket{10\cdots 0}=\ket{\psi_{s}^{(-)}},
\end{align}
by simply choosing a unitary within the even subspace and another within the odd subspace. Since these unitaries commute with total parity $P=\prod_{s=1}^q P_s$,
\begin{align}
\ket{G}
:=\Big(\bigotimes_{s=1}^{q}U_{s}^\dagger\Big)
\ket{\phi}
=\frac{\Pi_\sigma}
{\|\Pi_\sigma\ket{\psi}\|}
\bigotimes_{s=1}^{q}
\left(U_{s}^\dagger\ket{\psi_{s}}\right)
=\frac{\Pi_\sigma}
{\|\Pi_\sigma\ket{\psi}\|}
\bigotimes_{s=1}^{q}
\left(a_{s+}\ket{00\cdots0}
+a_{s-}\ket{10\cdots0}\right).
\end{align}
Let $j_s$ denote the first site of the block $A_{s}$, so that $c_{j_s}^\dagger$ creates a fermion in that site, and write $z_s=a_{s-}/a_{s+}$. Since
$a_{s+}+a_{s-}c_{j_s}^\dagger
=a_{s+}\exp(z_sc_{j_s}^\dagger)$, $\ket{G}$ is written explicitly as a Gaussian fermionic state
\begin{align*}
\ket{G}&\propto\Pi_\sigma\prod_{s=1}^{q}\Big(a_{s+}+a_{s-}c_{j_s}^\dagger\Big)\ket{0}^{\otimes N}
\propto \Pi_\sigma\exp\!\Big(\sum_{s=1}^q z_sc_{j_s}^\dagger+\mathcal{A}\Big)\ket{0}^{\otimes N}
=e^\mathcal{A}\Pi_\sigma \Big(1+\sum_{s=1}^qz_sc_{j_s}^\dagger\Big)\ket{0}^{\otimes N}
=e^\mathcal{A}\ket{v_\sigma},
\end{align*}
where $\mathcal{A}=\sum_{s<t}z_sz_t c_{j_s}^\dagger c_{j_t}^\dagger$, $\ket{v_+}=\ket{0}^{\otimes N}$, and $\ket{v_-}=\sum_{s=1}^q z_sc_{j_s}^\dagger \ket{0}^{\otimes N}$.

Furthermore, the unitaries $U_s$ can be chosen as circuits of nearest-neighbor parity-preserving gates of depth $\exp(\exp(O(\widetilde\beta)))$, as stated in Theorem~\ref{th:03}. Indeed, recall from Theorem~\ref{th:01} that each block state $\ket{\psi_s}$ is an MPS of bond dimension $\chi\leq 2^m\leq C_{\beta}$. Since the projector $\Pi_{s\pm}=(\1\pm P_s)/2$ is a sum of two product operators, the states $\ket{\psi_s^{(\pm)}}\propto\Pi_{s\pm}\ket{\psi_s}$ are MPSs of bond dimension at most $2\chi$ and of definite parity, i.e., fermionic MPSs~\cite{Kraus2010,Bultinck2017}. The sequential preparation of Ref.~\cite{Schon2005} then yields a circuit of $C_{\beta}$ gates acting on $O(\log\chi)$ neighboring modes that maps $\ket{00\cdots0}\mapsto\ket{\psi_s^{(+)}}$ and $\ket{10\cdots0}\mapsto\ket{\psi_s^{(-)}}$. We can choose these  gates to be parity preserving and to have determinant one in each parity sector. Each such gate then admits an exact decomposition into
$\poly(\chi)$ parity-preserving two-mode gates
by the constructive synthesis of
Ref.~\cite[Sec.~VIII~A]{Marvian2024}, so the depth is $O(C_{\beta}\poly(\chi))=\exp(\exp(O(\widetilde\beta)))$.

Finally, note that the unitaries $U_{s}$ were constructed as
unitaries on qubit blocks. For a block that does not cross
the Jordan--Wigner cut between sites $N$ and $1$, parity
preservation ensures that the strings outside the block
cancel, so $U_{s}$ is also a fermionic operator supported
on that block, but this is not true for the block crossing the cut. To handle this boundary block, one can simply perform a cyclic reordering such that the Jordan--Wigner cut lies between blocks. This change
of ordering introduces only on-site parity factors $Z_i=(-1)^{c_i^\dagger c_i}$ within the fixed total-parity sector, which can be absorbed into
the vectors $\ket{\psi_{s}}$ before choosing the corresponding
unitaries. 
\subsection{Sampling algorithm}
\label{app:ferm-sampling}

To sample the ensemble of Theorem~\ref{th:03}, we first need to sample the parity sector $\sigma$. For this, introduce an auxiliary qubit $0$
between qubits $1$ and $N$, and apply Theorem~\ref{th:02} to
$\widetilde H=\dyad{0}_0\otimes H_+ +\dyad{1}_0\otimes H_-$ with error $\delta$.
The Hamiltonians $H_+$ and $H_-$ differ only in terms crossing the bond between sites $1$ and $N$, so
$\widetilde H$ has interaction range at most $r+1$ and local
term norms at most one. Classically sampling the auxiliary qubit in the
computational basis 
selects $\sigma$ with probability $\tr(e^{-\beta H_\sigma})/({\tr(e^{-\beta H_+})+\tr(e^{-\beta H_-})}) +O(\delta)\geq (1+e^{2\beta r})^{-1}-O(\delta)$.
Once $\sigma$ is fixed, independently sample $\ket{\psi_{\sigma\mu}}$ by applying Theorem~\ref{th:02} to $H_\sigma$
 with error $\delta$, and accept with probability
\begin{equation}
\|\Pi_\sigma\ket{\psi_{\sigma\mu}}\|^2=\frac{1+\sigma\prod_s\bigl(|a_{s+}^{(\sigma\mu)}|^2-|a_{s-}^{(\sigma\mu)}|^2\bigr)}{2},
\end{equation}
where $a_{s\pm}^{(\sigma\mu)}$ are the amplitudes defined above for the sampled pair $(\sigma,\mu)$. Upon rejection, redraw both $\sigma$ and $\mu$. Upon
acceptance, output the Gaussian state and block unitaries
constructed above. Choosing sufficiently small
$\delta=O(\varepsilon (1+e^{2\beta r})^{-1})$ and stopping after
$O((1+e^{2\beta r})\log(1/\varepsilon))$ trials,
outputting the vacuum and identity unitaries on failure,
 gives trace-norm error $\varepsilon$ in time
$\poly(N,1/\varepsilon)$. Indeed, the overall ideal probability of accepting a trial
$(\sigma,\mu)$ is at least $\tr(e^{-\beta H})/({\tr(e^{-\beta H_+})+\tr(e^{-\beta H_-})})
\geq (1+e^{2\beta r})^{-1}$,
so conditioning on acceptance amplifies the
$O(\delta)$ sampling error by at most $O(1+e^{2\beta r})$,
while the probability of exhausting the allowed
trials contributes only $O(\varepsilon)$ to the trace-norm error.

\bibliography{references}
\end{document}